\documentclass[11pt]{article}
\usepackage{lmodern}
\usepackage[margin=3cm]{geometry}

\usepackage{lib/local}
\usepackage{iftex}

\usepackage[T1]{fontenc}
\usepackage[utf8]{\ifPDFTeX\else lua\fi inputenc}

\ifPDFTeX
\else
  \usepackage{fontspec}
  \usepackage{xunicode}
\fi

\usepackage{lib/text/url}

\usepackage{lib/layout/margintag}
\input{lib/layout/lineno}
\usepackage{lib/layout/notes}

\usepackage{mathtools}
\usepackage{amssymb}
\usepackage{amsfonts}

\def\eqsymbol{\mathrel{=}}
\def\defsymbol{\vcentcolon=}

\usepackage{iftex}

\ifPDFTeX
\usepackage{newunicodechar}

\AtBeginDocument{
\newunicodechar{ }{~}
\newunicodechar{¿}{\hbox{\textquestiondown}}
\newunicodechar{¡}{\hbox{\textexclamdown}}
\newunicodechar{≥}{\geq}
\newunicodechar{⇆}{\rightleftarrows}
\newunicodechar{⇁}{\rightharpoondown}
\newunicodechar{▹}{\rhd}
\newunicodechar{⇝}{\rightsquigarrow}
\newunicodechar{⋆}{\star}
\newunicodechar{×}{\times}
\newunicodechar{⟦}{\llbracket}
\newunicodechar{⟧}{\rrbracket}
\newunicodechar{⦅}{\left(}
\newunicodechar{⦆}{\right)}
\newunicodechar{ℝ}{\mathbb{R}}
\newunicodechar{ℕ}{\mathbb{N}}
\newunicodechar{ℚ}{\mathbb{Q}}
\newunicodechar{∞}{\infty}
\newunicodechar{⊕}{\oplus}
\newunicodechar{⊗}{\otimes}
\newunicodechar{≈}{\approx}
\newunicodechar{≡}{\equiv}
\newunicodechar{⅋}{\parr}
\newunicodechar{⊸}{\mathbin{\multimap}}
\newunicodechar{⊥}{\bot}
\newunicodechar{⊤}{\top}
\newunicodechar{←}{\leftarrow}
\newunicodechar{→}{\rightarrow}
\newunicodechar{↑}{\uparrow}
\newunicodechar{↓}{\downarrow}
\newunicodechar{∀}{\forall}
\newunicodechar{∃}{\exists}
\newunicodechar{∇}{\Nabla}
\newunicodechar{¬}{\neg}
\newunicodechar{∈}{\in}
\newunicodechar{⊆}{\subseteq}
\newunicodechar{⊂}{\subset}
\newunicodechar{⊇}{\supseteq}
\newunicodechar{⊃}{\supset}
\newunicodechar{∪}{\cup}
\newunicodechar{∩}{\cap}
\newunicodechar{⊔}{\sqcup}
\newunicodechar{⊓}{\sqcap}
\newunicodechar{∨}{\vee}
\newunicodechar{∧}{\wedge}
\newunicodechar{⊢}{\vdash}
\newunicodechar{◊}{\lozenge}
\newunicodechar{□}{\square}
\newunicodechar{⟝}{\vdash}
\newunicodechar{⟞}{\dashv}
\newunicodechar{⩴}{\vcentcolon\vcentcolon=}
\newunicodechar{≔}{\vcentcolon=}
\newunicodechar{α}{\alpha}
\newunicodechar{β}{\beta}
\newunicodechar{γ}{\gamma}
\newunicodechar{δ}{\delta}
\newunicodechar{ε}{\epsilon}
\newunicodechar{ζ}{\zeta}
\newunicodechar{η}{\eta}
\newunicodechar{θ}{\theta}
\newunicodechar{ι}{\iota}
\newunicodechar{κ}{\kappa}
\newunicodechar{λ}{\lambda}
\newunicodechar{μ}{\mu}
\newunicodechar{ν}{\nu}
\newunicodechar{ξ}{\xi}
\newunicodechar{π}{\pi}
\newunicodechar{ρ}{\rho}
\newunicodechar{σ}{\sigma}
\newunicodechar{τ}{\tau}
\newunicodechar{υ}{\upsilon}
\newunicodechar{φ}{\phi}
\newunicodechar{χ}{\chi}
\newunicodechar{ψ}{\psi}
\newunicodechar{ω}{\omega}
\newunicodechar{ο}{o}
\newunicodechar{Γ}{\Gamma}
\newunicodechar{Δ}{\Delta}
\newunicodechar{Θ}{\Theta}
\newunicodechar{Λ}{\Lambda}
\newunicodechar{Ξ}{\Xi}
\newunicodechar{Π}{\Pi}
\newunicodechar{Σ}{\Sigma}
\newunicodechar{Υ}{\Upsilon}
\newunicodechar{Φ}{\Phi}
\newunicodechar{Ψ}{\Psi}
\newunicodechar{Ω}{\Omega}
\newunicodechar{⋃}{\bigcup}

\newunicodechar{⊏}{\sqsubset}
\newunicodechar{⊐}{\sqsupset}
\newunicodechar{≤}{\leq}
\newunicodechar{≥}{\geq}
\newunicodechar{≺}{\prec}
\newunicodechar{≻}{\succ}
\newunicodechar{≼}{\preccurlyeq}
\newunicodechar{≽}{\succcurlyeq}
\newunicodechar{≼}{\preccurlyeq}
\newunicodechar{⊴}{\unlhd}
\newunicodechar{⊵}{\unrhd}
}

\else

\def\varnothing{∅}

\DeclareMathSymbol{α}{\mathalpha}{letters}{"0B}
\DeclareMathSymbol{β}{\mathalpha}{letters}{"0C}
\DeclareMathSymbol{γ}{\mathalpha}{letters}{"0D}
\DeclareMathSymbol{δ}{\mathalpha}{letters}{"0E}
\DeclareMathSymbol{ε}{\mathalpha}{letters}{"0F}
\DeclareMathSymbol{ζ}{\mathalpha}{letters}{"10}
\DeclareMathSymbol{η}{\mathalpha}{letters}{"11}
\DeclareMathSymbol{θ}{\mathalpha}{letters}{"12}
\DeclareMathSymbol{ι}{\mathalpha}{letters}{"13}
\DeclareMathSymbol{κ}{\mathalpha}{letters}{"14}
\DeclareMathSymbol{λ}{\mathalpha}{letters}{"15}
\DeclareMathSymbol{μ}{\mathalpha}{letters}{"16}
\DeclareMathSymbol{ν}{\mathalpha}{letters}{"17}
\DeclareMathSymbol{ξ}{\mathalpha}{letters}{"18}
\DeclareMathSymbol{π}{\mathalpha}{letters}{"19}
\DeclareMathSymbol{ρ}{\mathalpha}{letters}{"1A}
\DeclareMathSymbol{σ}{\mathalpha}{letters}{"1B}
\DeclareMathSymbol{τ}{\mathalpha}{letters}{"1C}
\DeclareMathSymbol{υ}{\mathalpha}{letters}{"1D}
\DeclareMathSymbol{φ}{\mathalpha}{letters}{"1E}
\DeclareMathSymbol{χ}{\mathalpha}{letters}{"1F}
\DeclareMathSymbol{ψ}{\mathalpha}{letters}{"20}
\DeclareMathSymbol{ω}{\mathalpha}{letters}{"21}
\DeclareMathSymbol{ο}{\mathalpha}{letters}{`o}

\DeclareMathSymbol{Γ}{\mathalpha}{letters}{"00}
\DeclareMathSymbol{Δ}{\mathalpha}{letters}{"01}
\DeclareMathSymbol{Θ}{\mathalpha}{letters}{"02}
\DeclareMathSymbol{Λ}{\mathalpha}{letters}{"03}
\DeclareMathSymbol{Ξ}{\mathalpha}{letters}{"04}
\DeclareMathSymbol{Π}{\mathalpha}{letters}{"05}
\DeclareMathSymbol{Σ}{\mathalpha}{letters}{"06}
\DeclareMathSymbol{Υ}{\mathalpha}{letters}{"07}
\DeclareMathSymbol{Φ}{\mathalpha}{letters}{"08}
\DeclareMathSymbol{Ψ}{\mathalpha}{letters}{"09}
\DeclareMathSymbol{Ω}{\mathalpha}{letters}{"0A}

\def\DefUnicodeCharacter#1{%
    \catcode\string`#1=\active
    \begingroup
    \lccode`~=\string`#1\relax
    \lowercase{\endgroup
      \protected\def~}}

\DefUnicodeCharacter{¿}{\hbox{\textquestiondown}}
\DefUnicodeCharacter{¡}{\hbox{\textexclamdown}}
\DefUnicodeCharacter{⅋}{\parr}
\DefUnicodeCharacter{□}{\square}
\DefUnicodeCharacter{⩴}{\vcentcolon\vcentcolon=}
\DefUnicodeCharacter{≔}{\vcentcolon=}
\DefUnicodeCharacter{⟝}{\vdash}
\DefUnicodeCharacter{⟞}{\dashv}
\DefUnicodeCharacter{⦅}{\left(}
\DefUnicodeCharacter{⦆}{\right)}
\DefUnicodeCharacter{ℝ}{\mathbb{R}}
\DefUnicodeCharacter{ℕ}{\mathbb{N}}
\DefUnicodeCharacter{ℚ}{\mathbb{Q}}
\DefUnicodeCharacter{∉}{\not\in}
\DefUnicodeCharacter{⇝}{\rightsquigarrow}
\DefUnicodeCharacter{⊏}{\sqsubset}
\DefUnicodeCharacter{⊐}{\sqsupset}
\DefUnicodeCharacter{⊴}{\unlhd}
\DefUnicodeCharacter{⊵}{\unrhd}

\DeclareMathSymbol{∂}{\mathord}{letters}{"40}
\DeclareMathSymbol{∞}{\mathord}{symbols}{"31}
\DeclareMathSymbol{′}{\mathord}{symbols}{"30}
\DeclareMathSymbol{∇}{\mathord}{symbols}{"72}
\DeclareMathSymbol{⊤}{\mathord}{symbols}{"3E}
\DeclareMathSymbol{⊥}{\mathord}{symbols}{"3F}
\DeclareMathSymbol{∀}{\mathord}{symbols}{"38}
\DeclareMathSymbol{∃}{\mathord}{symbols}{"39}
\DeclareMathSymbol{¬}{\mathord}{symbols}{"3A}

\DeclareMathSymbol{∐}{\mathop}{largesymbols}{"60}
\DeclareMathSymbol{⋁}{\mathop}{largesymbols}{"57}
\DeclareMathSymbol{⋀}{\mathop}{largesymbols}{"56}
\DeclareMathSymbol{⋂}{\mathop}{largesymbols}{"54}
\DeclareMathSymbol{⋃}{\mathop}{largesymbols}{"53}
\DeclareMathSymbol{∏}{\mathop}{largesymbols}{"51}
\DeclareMathSymbol{∑}{\mathop}{largesymbols}{"50}
\DeclareMathSymbol{⨂}{\mathop}{largesymbols}{"4E}
\DeclareMathSymbol{⨁}{\mathop}{largesymbols}{"4C}
\DeclareMathSymbol{⨀}{\mathop}{largesymbols}{"4A}
\DeclareMathSymbol{⨆}{\mathop}{largesymbols}{"46}
\DeclareMathSymbol{∧}{\mathbin}{symbols}{"5E}
\DeclareMathSymbol{∨}{\mathbin}{symbols}{"5F}
\DeclareMathSymbol{∩}{\mathbin}{symbols}{"5C}
\DeclareMathSymbol{∪}{\mathbin}{symbols}{"5B}
\DeclareMathSymbol{‡}{\mathbin}{symbols}{"7A}
\DeclareMathSymbol{†}{\mathbin}{symbols}{"79}
\DeclareMathSymbol{⊓}{\mathbin}{symbols}{"75}
\DeclareMathSymbol{⊔}{\mathbin}{symbols}{"74}
\DeclareMathSymbol{⋄}{\mathbin}{symbols}{"05}
\DeclareMathSymbol{∙}{\mathbin}{symbols}{"0F}
\DeclareMathSymbol{÷}{\mathbin}{symbols}{"04}
\DeclareMathSymbol{⊙}{\mathbin}{symbols}{"0C}
\DeclareMathSymbol{⊘}{\mathbin}{symbols}{"0B}
\DeclareMathSymbol{⊗}{\mathbin}{symbols}{"0A}
\DeclareMathSymbol{⊖}{\mathbin}{symbols}{"09}
\DeclareMathSymbol{⊕}{\mathbin}{symbols}{"08}
\DeclareMathSymbol{∓}{\mathbin}{symbols}{"07}
\DeclareMathSymbol{±}{\mathbin}{symbols}{"06}
\DeclareMathSymbol{∘}{\mathbin}{symbols}{"0E}
\DeclareMathSymbol{∖}{\mathbin}{symbols}{"6E}
\DeclareMathSymbol{⋅}{\mathbin}{symbols}{"01}
\DeclareMathSymbol{∗}{\mathbin}{symbols}{"03}
\DeclareMathSymbol{×}{\mathbin}{symbols}{"02}
\DeclareMathSymbol{⋆}{\mathbin}{letters}{"3F}
\DeclareMathSymbol{∝}{\mathrel}{symbols}{"2F}
\DeclareMathSymbol{⊑}{\mathrel}{symbols}{"76}
\DeclareMathSymbol{⊒}{\mathrel}{symbols}{"77}
\DeclareMathSymbol{∥}{\mathrel}{symbols}{"6B}
\DeclareMathSymbol{\mid}{\mathrel}{symbols}{"6A}
\DeclareMathSymbol{⊣}{\mathrel}{symbols}{"61}
\DeclareMathSymbol{⊢}{\mathrel}{symbols}{"60}
\DeclareMathSymbol{\nearrow}{\mathrel}{symbols}{"25}
\DeclareMathSymbol{\searrow}{\mathrel}{symbols}{"26}
\DeclareMathSymbol{\nwarrow}{\mathrel}{symbols}{"2D}
\DeclareMathSymbol{\swarrow}{\mathrel}{symbols}{"2E}
\DeclareMathSymbol{⟺}{\mathrel}{symbols}{"2C}
\DeclareMathSymbol{⇐}{\mathrel}{symbols}{"28}
\DeclareMathSymbol{⇒}{\mathrel}{symbols}{"29}
\DeclareMathSymbol{≤}{\mathrel}{symbols}{"14}
\DeclareMathSymbol{≥}{\mathrel}{symbols}{"15}
\DeclareMathSymbol{≻}{\mathrel}{symbols}{"1F}
\DeclareMathSymbol{≺}{\mathrel}{symbols}{"1E}
\DeclareMathSymbol{≈}{\mathrel}{symbols}{"19}
\DeclareMathSymbol{⪰}{\mathrel}{symbols}{"17}
\DeclareMathSymbol{⪯}{\mathrel}{symbols}{"16}
\DeclareMathSymbol{⊃}{\mathrel}{symbols}{"1B}
\DeclareMathSymbol{⊂}{\mathrel}{symbols}{"1A}
\DeclareMathSymbol{⊇}{\mathrel}{symbols}{"13}
\DeclareMathSymbol{⊆}{\mathrel}{symbols}{"12}
\DeclareMathSymbol{∈}{\mathrel}{symbols}{"32}
\DeclareMathSymbol{∋}{\mathrel}{symbols}{"33}
\DeclareMathSymbol{\gg}{\mathrel}{symbols}{"1D}
\DeclareMathSymbol{\ll}{\mathrel}{symbols}{"1C}
\DeclareMathSymbol{\not}{\mathrel}{symbols}{"36}
\DeclareMathSymbol{⟷}{\mathrel}{symbols}{"24}
\DeclareMathSymbol{←}{\mathrel}{symbols}{"20}
\DeclareMathSymbol{→}{\mathrel}{symbols}{"21}
\DeclareMathSymbol{⊸}{\mathbin}{AMSa}{"28}
\DeclareMathSymbol{∅}{\mathord}{AMSb}{"3F}
\DeclareMathSymbol{≡}{\mathrel}{symbols}{"11}
\DeclareMathSymbol{⌣}{\mathrel}{letters}{"5E}
\DeclareMathSymbol{⌢}{\mathrel}{letters}{"5F}

\DeclareMathDelimiter{↑}
   {\mathrel}{symbols}{"22}{largesymbols}{"78}
\DeclareMathDelimiter{↓}
   {\mathrel}{symbols}{"23}{largesymbols}{"79}
\DeclareMathDelimiter{↕}
   {\mathrel}{symbols}{"6C}{largesymbols}{"3F}
\DeclareMathDelimiter{⇑}
   {\mathrel}{symbols}{"2A}{largesymbols}{"7E}
\DeclareMathDelimiter{⇓}
   {\mathrel}{symbols}{"2B}{largesymbols}{"7F}
\DeclareMathDelimiter{⇕}
   {\mathrel}{symbols}{"6D}{largesymbols}{"77}
\DeclareMathDelimiter{⟩}
   {\mathclose}{symbols}{"69}{largesymbols}{"0B}
\DeclareMathDelimiter{⟨}
   {\mathopen}{symbols}{"68}{largesymbols}{"0A}

\fi

\usepackage{stmaryrd}
\usepackage{cmll}

\usepackage{lib/math/macros}
\usepackage{lib/math/connectors}
\usepackage{comment}

\usepackage[compact]{lib/disp}
\usepackage{lib/graphics/asy}\def\netstyle{smf4}

\usepackage{amsthm}
\newtheoremstyle{break}{\topsep}{\topsep}{\itshape}{}{\bfseries}{}{5pt}{}
\theoremstyle{break}
\newtheorem{theorem}{Theorem}
\newtheorem{corollary}{Corollary}
\newtheorem{definition}{Definition}
\newtheorem{lemma}{Lemma}
\newtheorem{example}{Example}

\newtheorem{proposition}{Proposition}
\usepackage{lib/colorscheme}
\usepackage[colorlinks=true,breaklinks,linktocpage]{hyperref}

\usepackage{lib/theory/rewriting}

\newcommand{\set}[1]{\left\{\,{#1}\,\right\}}
\newcommand{\setw}[2]{\left\{\,{#1}\;\middle|\;{#2}\,\right\}}

\newcommand{\msetunion}[1][]{\mathrel{#1{\tstrike{\cdot}{\cup}}}}

\newcommand{\msetempty}{\tstrike{\cdot}{\varnothing}}

\newcommand{\nodei}[2]{\ensuremath{\textsl{#2}\kern 1pt_{#1}}}

\newcommand{\redex}[2]{\ifmmode#1 -\penalty2000 #2\else#1 $-$ #2\fi}
\newcommand{\sredex}[2]{\ifmmode#1 →\penalty2000 #2\else#1 $→$ #2\fi}

\newcommand{\netvar}[1]{\mathord{\kern 1mm{#1}\kern 1mm}}
\newcommand{\netop}[1]{\mathbin{\kern 2mm{#1}\kern 2mm}}
\newcommand{\netempty}{\netvar{\varnothing}}

\usepackage{lib/theory/linear-logic}

\ifdefined\textls
\def\ntext#1{\textls[-15]{\textit{#1}}}
\else
\def\ntext#1{\textit{#1}}
\fi

\def\tbool{\mathsf{bool}}
\def\tnat#1{\mathsf{nat}_{#1}}
\def\tlist#1#2{\mathsf{list}_{#1}\mskip 1mu{#2}}
\def\tnatp#1{\mathsf{nat}^{#1}}
\def\tlistp#1#2{\mathsf{list}^{#1}\mskip 1mu{#2}}

\def\tcomp#1#2{\mathsf{cmp}_{#1}\mskip 1mu{#2}}
\def\fh#1{\lfloor\kern -1pt\frac{#1}{2}\kern -1pt\rfloor}
\def\ch#1{\lceil\kern -1pt\frac{#1}{2}\kern -1pt\rceil}

\def\sc#1{#1}
\def\pc#1{\bar{#1}}

\def\stime#1{\mathfrak{Time}_{#1}}
\def\sspace#1{\mathfrak{Space}_{#1}}
\def\sst#1{\mathfrak{Space–Time}_{#1}}
\def\ptime#1{\bar{\mathfrak{Time}}_{#1}}
\def\pspace#1{\bar{\mathfrak{Space}}_{#1}}
\def\pst#1{\bar{\mathfrak{Space–Time}}_{#1}}

\DeclareMathOperator{\bigO}{O}
\newcommand*{\tpkt}{\rlap{$\;$.}}
\newcommand*{\tkom}{\rlap{$\;$,}}
\newcommand{\defsym}{\mathrel{:=}}

\newcommand*{\FS}{\mathcal{S}}
\newcommand*{\RS}{\mathcal{R}}
\newcommand*{\activepair}[2]{(#1,#2)}
\newcommand*{\netsize}[1]{{\#{#1}}}

\newcommand*{\blank}{{\sqcup}}
\newcommand*{\lmark}{{\vdash}}
\newcommand*{\rmark}{{\dashv}}
\newcommand*{\lmove}{\text{L}}
\newcommand*{\rmove}{\text{R}}
\newcommand*{\tmstart}{s}
\newcommand*{\tmstop}{t}
\newcommand*{\code}[1]{{\bar #1}}
\newcommand*{\ceil}[1]{\lceil #1 \rceil}
\newcommand*{\yields}[1]{\mathrel{\smash{\xrightarrow[\mathcal{T}]{#1}}}}
\newcommand*{\N}{\mathbb{N}}
\newcommand*{\rsize}[1]{\lVert #1 \rVert}

\def\todo#1{}
\newif\ifextended
\extendedtrue

\def\subparagraph#1{\smallskip{\noindent\normalsize\bfseries#1.\hskip 0.4ex plus .2ex\relax}\nobreak}

\title{The Complexity of Interaction}
\author{Stéphane Gimenez, Georg Moser}

\usepackage[backend=bibtex,style=alphabetic,sorting=ynt,hyperref=auto]{biblatex}
\bibliography{bib/refs}

\usepackage{parskip}

\begin{document}

\maketitle

\begin{abstract}
In this paper, we analyze the complexity of functional programs written in the interaction-net computation model, an asynchronous, parallel and confluent model that generalizes linear-logic proof nets.
Employing user-defined \emph{sized} and \emph{scheduled} types, we certify concrete time, space and space–time complexity bounds for both sequential and parallel reductions of interaction-net programs by suitably assigning complexity potentials to typed nodes.
The relevance of this approach is illustrated on archetypal programming examples.
The provided analysis is precise, compositional and is, in theory, not restricted to particular complexity classes.
\end{abstract}

\section{Introduction}

Complexity analysis provides bounds on the amount of resources required for a computation (chiefly time or space) relative to an input size.

Suppose that a function call $f(x)$ executes in time $\bigO(|x|)$ and $g(x)$ executes in time $\bigO(2^{|x|})$.
What is the time complexity of the composition $g \circ f$ of these two functions?
Is it perhaps $\bigO(2^{|x|})$?
The real answer is: we don't know.
It depends on the size of $f(x)$, which could be logarithmic, linear, or even exponential (for example in a parallel computation model), among other possibilities, with respect to the size of $x$.
The complexity of the composition $g \circ f$ varies consequently:
\display{\tb{
\begin{tabular}{l@{\hspace{3ex}}l}
  $\square\square\square\square \reducts[log] \square\square \reducts[g] a$ & $\bigO(|x|)$\\[5pt]
  $\square\square\square\square \reducts[id] \square\square\square\square \reducts[g] b$ & $\bigO(2^{|x|})$\\[5pt]
  $\square\square\square\square \reducts[exp] \square\square\square\square\square\square\square\square\square\square\square\square\square\square\square\square \reducts[g] c$ & $\bigO(2^{2^{|x|}})$\\
\end{tabular}
}}

Complexity analysis in the traditional sense is therefore not compositional.
To ensure compositionality, in this paper, we analyze properties of outputs, relying on user-defined types which have been enriched with size and, for parallel reductions, timing information.

We base our study on interaction nets, which provide a tangible cost model for both sequential \emph{and} parallel reductions.
Interaction nets have been used as an execution platform for functional programs~\cite{Mackie:2000,Pinto:2001,Mackie04}, as a conceptual device for the optimal implementation of the $\lambda$-calculus~\cite{Lamping:1990,linear-logic-without-boxes--g-ghontier+m-abadi+j-j-levy,AspertiGN96}, as a general purpose higher-order language~\cite{FMSW:2009,thesis--s-gimenez}, and as a model of distributed computation exemplified by asynchronous abstract hardware~\cite{Lippi:2009}.

Interaction nets provide a Turing-complete computation model, allow a complexity analysis of sequential reductions of (higher-order) functional programs and allow a reasonably painless extension to parallel reductions, an area typically ignored in the literature (see~\cite{HS:2015} for the exception to the rule).
Moreover, as interaction nets incorporate distributed computation, they are of relevance for the study of modern hardware platforms.
Hence, a static complexity analysis of interaction nets is also of interest in its own right.

In addition to providing an analysis for space and time complexity separately, we emphasize the study of \emph{space–time} complexities, i.e., space occupation as a function of time.
On the one hand this is a neat technical tool, as certification of precise time or space complexity bounds for sequential and parallel reductions come as very easy corollaries.
On the other hand, an accurate prediction of the space–time complexity could prove useful in the application of interaction-net technology in the context of distributed computation.

For sequential reductions, user-defined \emph{sized types} allow to keep track of arbitrarily chosen size measures for intermediate results.
From this we obtain a compositional analysis for the space–time complexity of interaction nets which relies on a suitable assignment of potentials to nodes (Theorem~\ref{t:sequential}).
A practical difficulty is the combination of these potentials, because the juxtaposition of two terms that are tied to sequential space–time complexities $f$ and $g$ is a complex convolution $h(t) = \max_{u+v=t} f(u) + g(v)$ where $u$ and $v$ denote the respective times allocated to the reduction of each term.
%
We draw attention to the fact that sequential computation in interaction nets slightly deviates from the standard notion.
Traditionally in a sequential deterministic computation, the instruction to be performed next is predefined for the single execution thread.
However, interaction-net systems rely on the diamond property to guarantee confluence and thus allow “don't-care non-determinism”, which our complexity analysis can handle.

To address parallel reductions, we use timing annotations which can be combined with size annotations to control the schedule of the computation.
The resulting \emph{scheduled types} express guaranteed or (for inputs) expected time limits on data availability.
Based on this we obtain again a precise analysis of space–time complexities for parallel reductions (Theorem~\ref{t:parallel}).
While typing schemes become more involved in the context of parallel reductions, the space–time complexity of a juxtaposition is in this case easily expressed as a sum $h(t) = f(t) + g(t)$.

The remainder of this paper is structured as follows.
Section~\ref{s:interaction-nets} briefly surveys interaction nets and outlines our motivating examples.
\ifextended
In Section~\ref{s:implementation} we link the interaction-net cost model to more traditional notions of time and space complexity.
\fi
In Section~\ref{s:sized-types} we define sized types with a semantic complexity analysis in mind.
These are employed in Section~\ref{s:sequential-analysis}, in which we state and prove our first main theorem concerning computational complexities of sequential reductions.
Section~\ref{s:scheduled-types} introduces scheduled types.
We use them in Section~\ref{s:parallel-analysis} to establish our second main theorem concerning computational complexities of parallel reductions.
In Section~\ref{s:higher-order} we show how our results can be used to analyze weak sequential reductions of higher-order programs.
In Section~\ref{s:related} we discuss related work.
Finally we conclude in Section~\ref{s:conclusion}, where we also report on future work.

\def\cscheme#1#2{\mb{#2\,{\begin{matrix}#1\end{matrix}}}}
\def\abracket#1{(#1)}
\def\boxcode#1{\parbox{300pt}{\vskip 2pt\tt\footnotesize #1}}

\section{Interaction Nets and their Parallel Reduction}
\label{s:interaction-nets}

The reader is assumed to be familiar with interaction nets as described by Lafont in \cite{interaction-nets--y-lafont}.
Various formal definitions of interaction nets can be found elsewhere in the literature \cite{thesis--d-mazza, thesis--l-vaux, Perrinel:2014}, but an intuitive understanding of the underlying graph-rewriting technology should be sufficient in the context of this paper.

An \emph{interaction-net system} is a pair $(\FS,\RS)$ consisting of a set of \emph{symbols} $\FS$ used to label \emph{nodes} and an associated set of \emph{reduction rules} $\RS$.
Each symbol has an integer \emph{arity} that dictates the number of \emph{auxiliary ports} which nodes labeled with this symbol possess.
Additionally, each node possesses exactly one distinguished \emph{principal port}.
Graphically, in the following set of labeled nodes, which will be used to represent constructors ($\ntext{zero}$, nullary, $\ntext{succ}$, unary) and operations ($\ntext{add}$ and $\ntext{mul}$, binary) over natural numbers, auxiliary ports are located at the top and principal ports are located at the bottom.
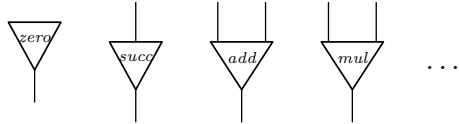
\display{
  \mb{
\input{asy/a407ccfe/\@netformat-\netstyle}
    \quad
\input{asy/e22a4b99/\@netformat-\netstyle}
    \quad
\input{asy/cd96dfbc/\@netformat-\netstyle}
    \quad
\input{asy/72527180/\@netformat-\netstyle}
    \quad \dots
  }
}

An \emph{interaction net} is a graph built from such nodes, where each port can be connected to one other port by one wire.
Unconnected ports of a net are called free ports and are collectively referred to as the \emph{interface} of this net.
By design, interaction-net redexes consist of two nodes only (whose labels are represented abstractly by $⋆_1$ and $⋆_2$ in the following picture) that are connected through their principal ports.
Reduction rules reduce them in context to a given net:
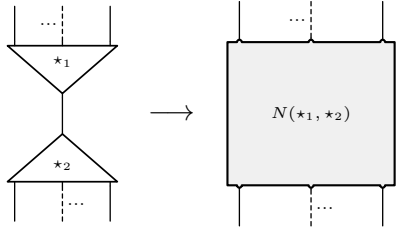
\display{
  \opreducts{
\input{asy/46cfe2f2/\@netformat-\netstyle}
    \and
\input{asy/20bad5fa/\@netformat-\netstyle}
  }
}

To ensure determinism when firing a single redex, only one reduction rule is allowed per symbol pair and the net $N(⋆_1, ⋆_2)$ is assumed symmetric if $⋆_1 = ⋆_2$.
Because redexes cannot overlap, the reduction enjoys the diamond property and can be parallelized easily; normal forms are unique.
Examples can be found in \cite{interaction-nets--y-lafont}.

Interaction nets form a simple and realistic computation model as the locating, recording and firing of one redex can all be implemented in constant time and constant space on standard computer architectures.

\subparagraph{Typed Interaction}
Types are syntactic expressions built from base types, which may expect other types as arguments, and polymorphic variables denoted by $A$, $B$, etc. which can be instantiated at will.

As a running example, we will consider the following set of typed primitives for natural numbers, abstractions and lists where $\mathsf{nat}$ (nullary), $⊸$ (infix binary), $\mathsf{list}$ (unary) and $\oc$ (unary, called \emph{exponential type} and used to mark polymorphically replicable data) are used as base types.
This includes most of the essential ingredients of a typical functional programming language.
Yet, our methodology is not restricted to this set of symbols and associated reduction rules.
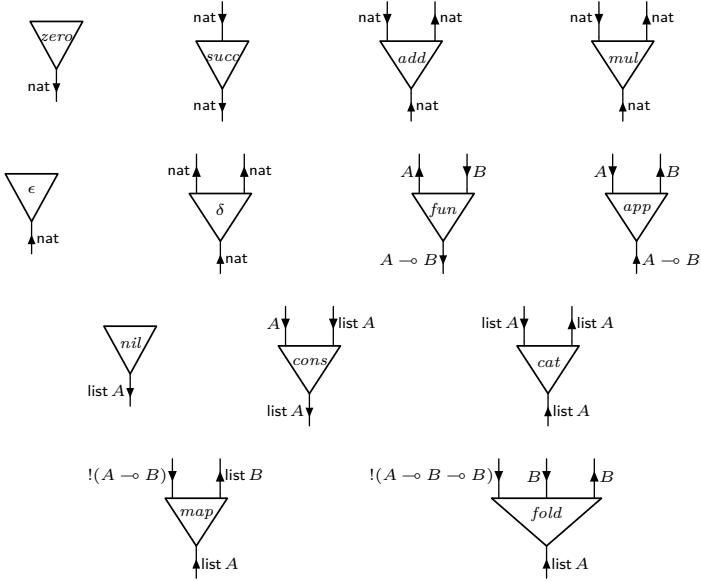
\display{
\input{asy/9730926c/\@netformat-\netstyle}
  \and
\input{asy/ce99b99c/\@netformat-\netstyle}
  \and
\input{asy/701b85e9/\@netformat-\netstyle}
  \and
\input{asy/35cb7e1f/\@netformat-\netstyle}
  \also
\input{asy/c5096bf4/\@netformat-\netstyle}
  \and
\input{asy/8c6db1c6/\@netformat-\netstyle}
  \and
\input{asy/7545785a/\@netformat-\netstyle}
  \and
\input{asy/25791321/\@netformat-\netstyle}
  \also
\input{asy/c449de53/\@netformat-\netstyle}
  \and
\input{asy/3c7a97d7/\@netformat-\netstyle}
  \and
\input{asy/f1c3aedc/\@netformat-\netstyle}
  \also
\input{asy/1a45b9f4/\@netformat-\netstyle}
  \and
\input{asy/8352ba05/\@netformat-\netstyle}
}

In the above \emph{typing schemes} (one has been provided for every symbol), ports have been oriented and attributed a type.
Typically, nodes used as type constructors may admit inputs as auxiliary ports and they output an object of the corresponding type on their principal port.
Other nodes can be considered as functions that seek to interact with their first arguments (provided as inputs on their principal ports), may expect more arguments as additional inputs on auxiliary ports and may output any number of results on remaining auxiliary ports.


In a \emph{typed interaction net}, wires and free ports are assigned types which must be instances of the types from the typing schemes associated to the symbols that label the nodes to which they are attached.
In other words, instances of typing schemes which have been assigned to nodes have to unify on every wire, with matching type orientations:
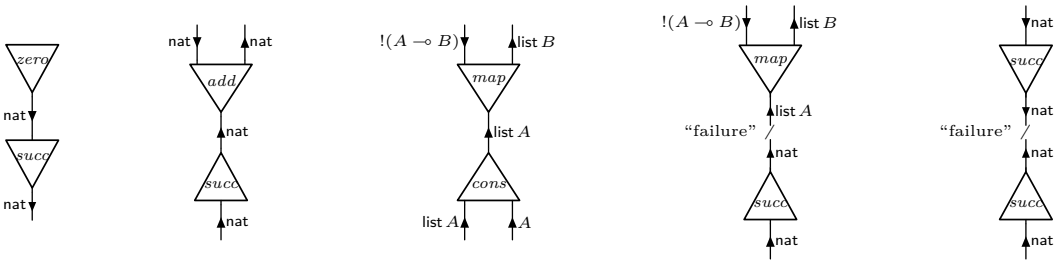
\display{
\input{asy/b9457ade/\@netformat-\netstyle}
  \and
\input{asy/ca7544b5/\@netformat-\netstyle}
  \and
\input{asy/ca422ab4/\@netformat-\netstyle}
  \also
\input{asy/447171da/\@netformat-\netstyle}
  \and
\input{asy/726a296c/\@netformat-\netstyle}
}

Reduction rules for addition of natural numbers are usually defined as follows:
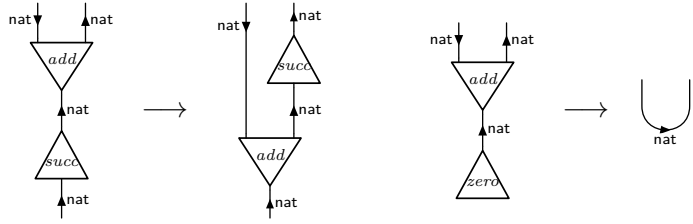
\display{
  \opreducts{
\input{asy/ca7544b5/\@netformat-\netstyle}
    \and
\input{asy/62a58d26/\@netformat-\netstyle}
  }
  \and
  \opreducts{
\input{asy/9e8bba9e/\@netformat-\netstyle}
    \and
\input{asy/c9ffad70/\@netformat-\netstyle}
  }
}

In a \emph{typed interaction-net system}, symbols are provided together with typing schemes and all reduction rules $L \reducts R$ are assumed to be typed and to preserve the types of their interfaces.
This means that $L$ and $R$ are typed interaction nets and the typing of the interface of $L$ matches the one of $R$.
This assumption entails a subject reduction property: any reduct of a typed net can be typed in a way that preserves the typing of its interface.

\subparagraph{Replication}
Nodes $\ntext{fun}$ and $\ntext{app}$ together with the following reduction rule suffice to encode the linear $\lambda$-calculus; see \cite{thesis--s-gimenez}.
\display{
  \opreducts{
\input{asy/95c86ec3/\@netformat-\netstyle}
    \and
\input{asy/54125809/\@netformat-\netstyle}
  }
}

In order to ease the understanding of this rule, virtual body, argument and continuation passed to the function can be conceptualized as a context:
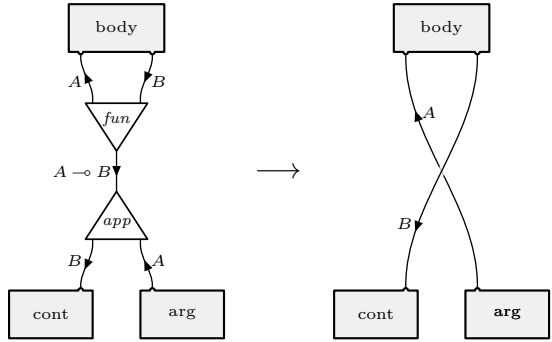
\display{
  \opreducts{
\input{asy/281a75e5/\@netformat-\netstyle}
    \and
\input{asy/1a1109c2/\@netformat-\netstyle}
  }
}

In order be used as an expressive higher-order language similar to the full $\lambda$-calculus, polymorphic replication is necessary.
Various implementations exist in interaction nets; some rely on an infinite family of sharing nodes (as in sharing graphs \cite{linear-logic-without-boxes--g-ghontier+m-abadi+j-j-levy}); other use special devices called \emph{boxes} which, strictly speaking, are not interaction-net nodes.
When associated to weak reduction rules, the reduction of boxes admits the diamond property and moreover corresponds quite closely to the weak $\beta$-reductions implemented by usual functional programming languages.
All the content presented in this paper is compatible with the use of such boxes.
Namely, we chose to work with \emph{functorial promotion boxes} \cite{llwiki,functorial-boxes-in-string-diagrams--p-a-mellies}, which are parameterized by a net (represented below as $N$), together with the usual \emph{weakening} ($w$), \emph{contraction} ($c$), \emph{dereliction} ($?$) and \emph{digging} ($¿$) nodes from linear logic.
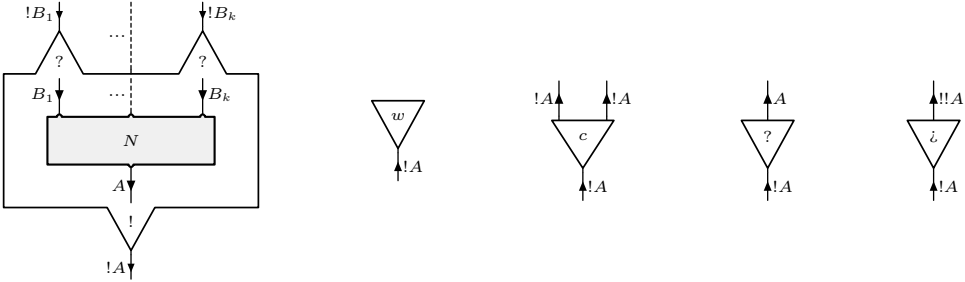
\display{
\input{asy/ea62c1a6/\@netformat-\netstyle}
  \also
\input{asy/1fb38d82/\@netformat-\netstyle}
  \and
\input{asy/925bd1f9/\@netformat-\netstyle}
  \and
\input{asy/9b4debfa/\@netformat-\netstyle}
  \and
\input{asy/63bbed46/\@netformat-\netstyle}
}

In a weak setting, nets that parameterize boxes are not reduced internally.
Boxes are reduced externally and non-closed boxes ($k > 0$) may only merge with other boxes until they are eventually closed (some types are omitted for clarity):
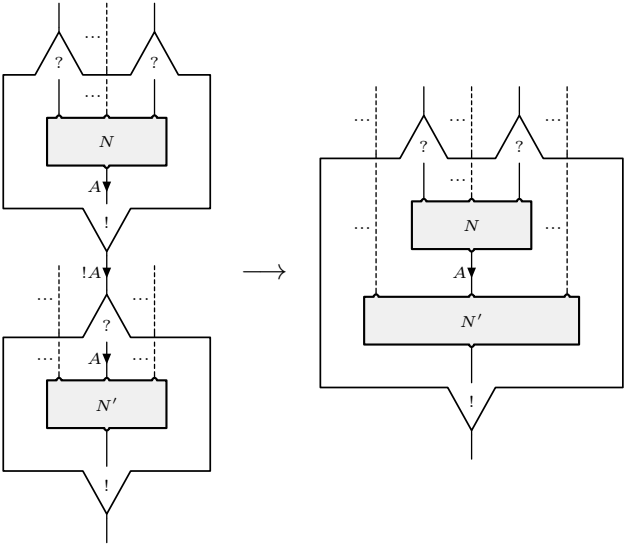
\display{
  \opreducts{
\input{asy/241d9659/\@netformat-\netstyle}
    \and
\input{asy/e4756000/\@netformat-\netstyle}
  }
}
Closed boxes ($k = 0$) can be:
\begin{itemize}
  \item erased by \emph{weakenings}:
    \display{
      \opreducts{
\input{asy/3fed7623/\@netformat-\netstyle}
        \and
        \mb{\netempty}
      }
    }
  \item duplicated by \emph{contractions}:
    \display{
      \opreducts{
\input{asy/f47fcf7d/\@netformat-\netstyle}
        \and
\input{asy/801585d3/\@netformat-\netstyle}
      }
    }
  \item opened by \emph{derelictions}:
    \display{
      \opreducts{
\input{asy/4e9e6593/\@netformat-\netstyle}
        \and
\input{asy/0506499f/\@netformat-\netstyle}
      }
    }
  \item and doubled by \emph{diggings}:
    \display{
      \opreducts{
\input{asy/da969d82/\@netformat-\netstyle}
        \and
\input{asy/8233991c/\@netformat-\netstyle}
      }
    }
\end{itemize}

As an illustration, Figure~1 depicts the 5-step fully parallel reduction (i.e. all available redexes are fired together at once) that reduces to normal form the mapping of the successor constructor to the list of natural numbers $[0, 0]$, according to the following reduction rules:
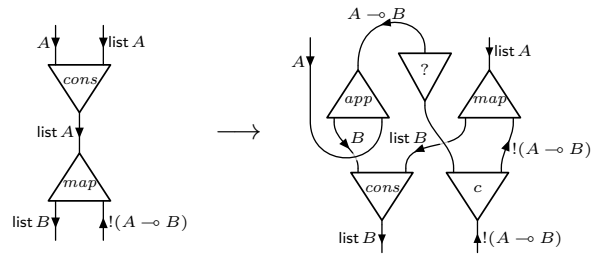
\display{
  \opreducts{
\input{asy/c2bd5eba/\@netformat-\netstyle}
    \and
\input{asy/66d826e5/\@netformat-\netstyle}
  }
  \and
  \opreducts{
\input{asy/38a15197/\@netformat-\netstyle}
    \and
\input{asy/8e7750f4/\@netformat-\netstyle}
  }
}

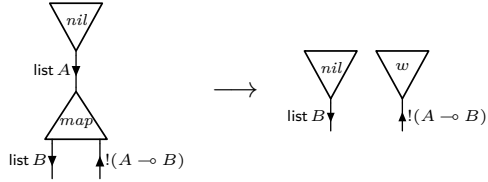
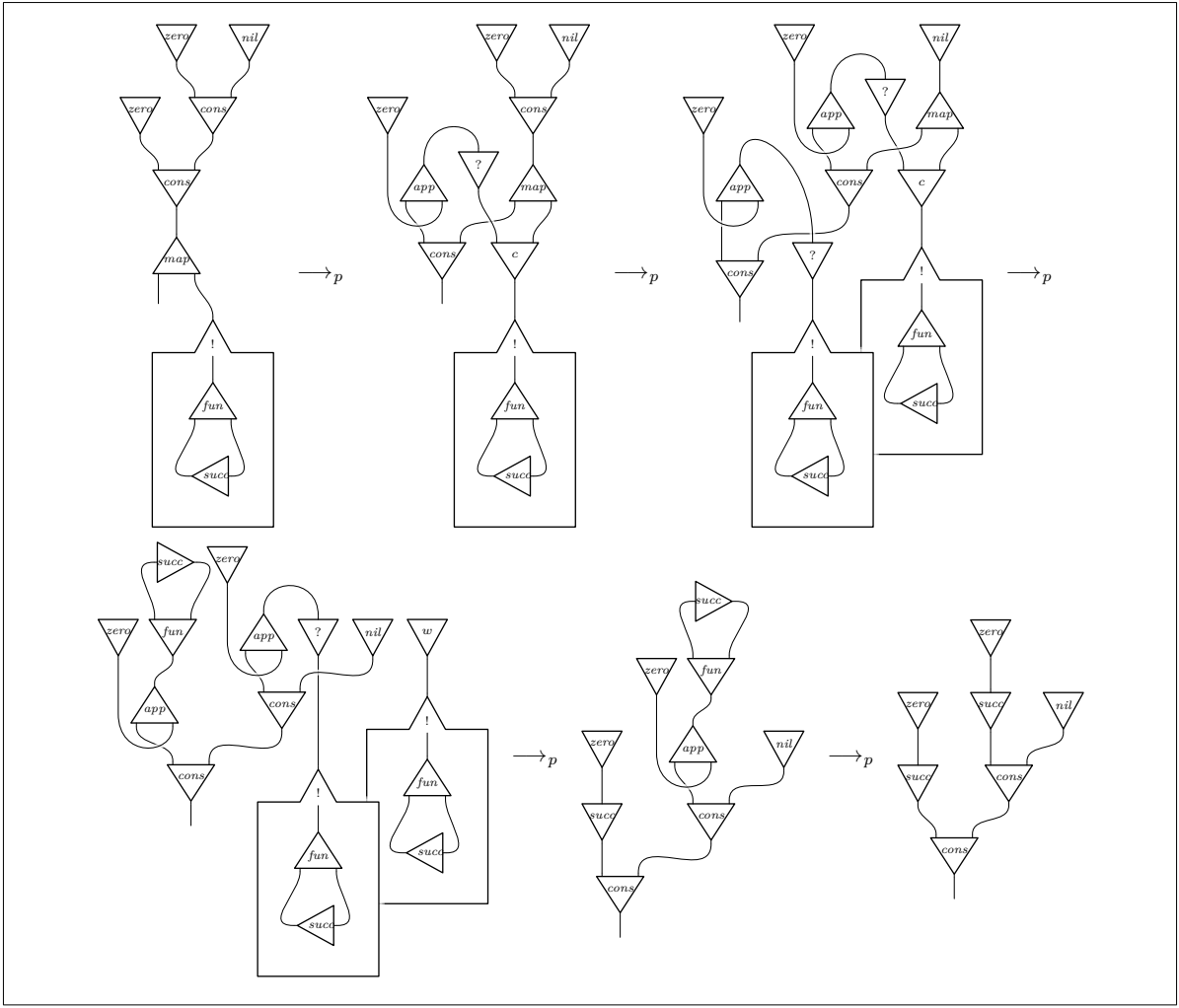
\figuredisplay[1.15]{Fully parallel reduction of a functional program}{
  \op{\mb{\reducts_p}}{
\input{asy/0a061120/\@netformat-\netstyle}
    \and
\input{asy/233ea477/\@netformat-\netstyle}
    \and
\input{asy/91eca027/\@netformat-\netstyle}
    \and
    {}
  }
  \op{\mb{\reducts_p}}{
\input{asy/b66ccea7/\@netformat-\netstyle}
    \and
\input{asy/fffe3d77/\@netformat-\netstyle}
    \and
\input{asy/6ccdd9bf/\@netformat-\netstyle}
  }
}

\subparagraph{Timed Interaction}
We consider a generalization of interaction-net systems to timed reduction rules $L \reducts[d] R$, in which the label $d ≥ 0$ denotes a time duration in a chosen time domain $T$, which can either be discrete or continuous.

We define \emph{timed sequential reduction} $\reducts_s$ as the extension of the timed reduction rules generated by the following transitivity property:
$N \reducts[t]_s M$ if there exists $t_1$, $t_2$ and a net $P$ such that $t = t_1 + t_2$ and $N \reducts[t_1]_s P$ and $P \reducts[t_2]_s M$.
The timed sequential reduction satisfies a timed diamond property: if $N \reducts[t_1]_s P_1$ and $N \reducts[t_2]_s P_2$ then there exist $M$ such that $P_1 \reducts[t_2]_s M$ and $P_2 \reducts[t_1]_s M$.
This ensures that all reductions to normal form have the same duration.

The definition of \emph{timed parallel reduction} $\reducts_p$ requires labeling every instance of a redex $r$ with a counter value $c(r)$ (defaulting to $0$ initially) that will always be less than the time $d(r)$ associated to the corresponding reduction rule.
We perform the parallel reduction $N \reducts[t]_p M$ of duration $t$ of a net $N$ (assumed not to be in normal form) by increasing all redex counters by $t$ if all counters satisfy $c(r) + t < d(r)$. Otherwise, by firing simultaneously all redexes which minimize $d(r) - c(r)$, then increasing other counters by the obtained minimal value $t_0$, and pursuing recursively the reduction for a duration of $t - t_0$.
The parallel reduction $\reducts[t]_p$ is strictly deterministic and satisfies the same transitivity property as the sequential reduction.
\todo{example!}

In this paper, for simplicity, we will only consider examples where reduction rules (with the exception of type-conversion rules which will be defined later) are arbitrarily assigned unitary time durations, and correspond to standard interaction-net systems.
The theory however allows one to take into account various implementation details since it is general and applicable to arbitrary durations as well.

\subparagraph{Cost Model}
Given a timed reduction $\reducts$ and a notion of space occupation $|\cdot|$ for nets (typically, the number of nodes) in a chosen space domain $S$, we say that $N$ admits:
\begin{itemize}
\item $τ ∈ T$ as a time bound if whenever $N \reducts[t] M$, $t ≤ τ$.
\item $σ ∈ S$ as a space bound if whenever $N \reducts[t] M$, $|M| ≤ σ$.
\item $γ : T \rightarrow S$ as a space–time bound if whenever $N \reducts[t] M$, $|M| ≤ γ(t)$.
\end{itemize}

In the following sections, we provide methods to compute such bounds for any net in a given interaction-net system, both for the timed sequential reduction $\reducts_s$ and the timed parallel reduction $\reducts_p$.
These methods rely on user-defined assignments of potentials to typed nodes, which must be provided together with the defined interaction-net system.
We will illustrate the use of these methods by providing such assignments for all the nodes we have introduced.

\ifextended
Related to this cost model, a notion of complexity can be defined for nets which expect inputs.
For example, the time complexity $\mathfrak{T}_N(n)$ of a net $N$ with respect to the size $n$ of one of its input can be defined as the maximum of the time bounds obtained for the following compositions of nets, where $I$ is a net taken in a given set $\mathcal{I}_n$ of inputs (in normal form) of size $n$.
\display{
\input{asy/8406a608/\@netformat-\netstyle}
}

In the following developments, types themselves will be annotated with sizes.
Consequently, the set $\mathcal{I}_n$ can naturally be chosen as the set $\mathcal{N}_{A(n)}$ of closed normal-form inhabitants of the corresponding sized type $A(n)$.
\display{
  \releq{
    \mb{\mathfrak{T}_N(n)}
    \and
    \mb{\max_{I∈\mathcal{N}_{A(n)}}\: τ\left(
\input{asy/242d525e/\@netformat-\netstyle}
    \right)
    }
  }
}

This notion of complexity generalizes easily to nets expecting multiple inputs.
Moreover, because sized types may freely be parameterized by any number of size variables, this notion also generalizes to inputs measured in more than a single size parameter.
In particular, some of our examples are concerned with the complexity of operations on lists whose lengths are bounded by $n$ containing elements whose size is bounded by $m$.
\fi

\ifextended
  \section{An Implementation of Interaction-net Reductions}
\label{s:implementation}

In this section, we link the above cost model for interaction-net systems with more traditional cost models.
We restrict our attention to sequential reductions on finite interaction-net systems without the use of boxes.

With respect to sequential reduction on an interaction model (without boxes), we show that interaction nets form a \emph{reasonable}~\cite{Boas:1990} cost model for time.
On the one hand computations on Turing machines can be simulated step by step with interaction nets. On the other, a computation of a net $N$ can be computed on a Turing machine in polynomial time with respect to the original reduction length.
The latter follows by a straightforward encoding of the net as an adjacency list of nodes.
However, this (obvious) encoding is not sufficient to obtain a constant-factor overhead in space, but requires a logarithmic overhead.
We recall the \emph{invariance thesis}~\cite{Boas:1990}:
\begin{quote}
  \emph{Reasonable machines can simulate each other within polynomially bounded
    overhead in time and a constant-factor overhead in space.}
\end{quote}

Firstly, we show that any computation on a Turing machine (TM for short) can be simulated step by step within a suitably defined interaction-net system.
To simplify the encoding, we restrict our attention to single-tape TMs, whose tapes grow infinitely in one direction.
TMs are represented as quintuples $\mathcal{T} = (Q,\Sigma,\delta,\tmstart,\tmstop)$, where $Q$ denotes the finite set of states (including the start state $\tmstart$ and halting state $\tmstop$);
$\Sigma$ denotes the (tape) alphabet, including a dedicated symbol for the left-end marker $\lmark$ and the blank symbol $\blank$;
finally $\delta$ denotes the transition function.
The precise definition is not of relevance here, as these definitions are standard, cf.~\cite{Kozen:1997}.

\begin{definition}
Let $\mathcal{T} = (Q,\Sigma,\delta,\tmstart,\tmstop)$ be a TM.
We define the interaction-net system $(\FS,\RS)$ \emph{corresponding} to $\mathcal{T}$.
For each letter $a \in \Sigma$, there exists a symbol $\code{a} \in \FS$ and for each state $p \in Q$, we make sure that $\{\code{p}_\lmove, \code{p}_\rmove\} \subseteq \FS$.
Furthermore, we add a special symbol, $\rmark \in \FS$, which encodes the end of the tape.
The arity of all the symbols except $\lmark$ and $\rmark$ is one, while the arity of the endmarkers is zero.
The set of reduction rules $\RS$ is defined as follows:
If $\delta(p,a) = (q,b,\rmove)$, then $\RS$ includes the reduction rules:
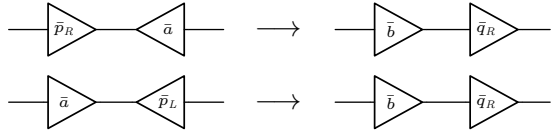
\display{
  \opreducts{
\input{asy/9d22fcbc/\@netformat-\netstyle}
    \and
\input{asy/5b0d0afa/\@netformat-\netstyle}
  }
  \opreducts{
\input{asy/62b0c9bd/\@netformat-\netstyle}
    \and
\input{asy/93cdbd1a/\@netformat-\netstyle}
  }
}
Otherwise, if $\delta(p,a) = (q,b,\lmove)$, then $\RS$ includes the reduction rules:
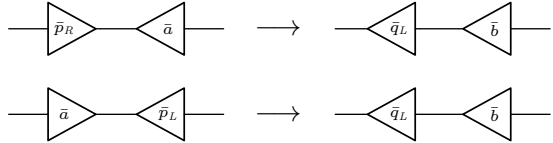
\display{
  \opreducts{
\input{asy/9d22fcbc/\@netformat-\netstyle}
    \and
\input{asy/cc867f63/\@netformat-\netstyle}
  }
  \and
  \opreducts{
\input{asy/62b0c9bd/\@netformat-\netstyle}
    \and
\input{asy/cc867f63/\@netformat-\netstyle}
  }
}

In addition, $\RS$ contains the following reduction rule to handle the right endmarker:
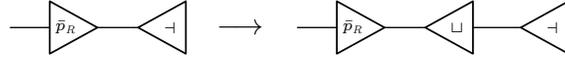
\display{
  \opreducts{
\input{asy/61c357a5/\@netformat-\netstyle}
    \and
\input{asy/536b4287/\@netformat-\netstyle}
  }
}

\end{definition}

Configurations of a TM are denoted as triples $(p,y,n)$, where $p \in Q$, $y \in \Sigma^\ast$, and $n \in \N$ denotes the position of the head.
For example $(\tmstart,\lmark x \blank^\infty,0)$ denotes the initial configuration of $\mathcal{T}$ with the word $x$ on the single tape, and the tape head pointing to the left endmarker~$\lmark$.

Let $\mathcal{T}$ be a TM, its single tape is represented as a string of cells, whose ends are marked by the (codes of the) endmarkers and the position of the head is represented by a cell labelled with the current state.
We say that a net $N$ \emph{encodes} a configuration $(p,y,n)$, if the word $y$ can be read off from $N$ starting from the cell $\lmark$ and continuing to the cell $\rmark$, where cells representing states are ignored.
Furthermore the active pair of $N$ occurs at the $n$\textsuperscript{th} position (counting only symbols of the alphabet).

\begin{lemma}
Let $\mathcal{T}$ be a TM and let $(\FS,\RS)$ be the corresponding interaction-net system.
Suppose $(\tmstart,\lmark x \blank^\infty,0) \yields{\ast} (\tmstop,y,n)$.
Then there exists a net $N$ built on $\FS$ such that $N \reducts[\ast] M$, where $N$ encodes the initial configuration and $M$ encodes the final configuration.
\end{lemma}
\begin{proof}
By construction.
\end{proof}

The next proposition is a direct consequence of the lemma.
\begin{proposition}
Let $\mathcal{T} = (Q,\Sigma,\delta,\tmstart,\tmstop)$ be a TM computing a partial function $f$
with range $\Sigma^\ast$ and let $(\FS,\RS)$ denote the corresponding interaction-net system.
Then for each sequence of words $x_1,\dotsc,x_n$ (the arguments of $f$), there exists a net $N$ based on $\FS$ such that the computation of $f$ on $\mathcal{T}$ can be simulated step by step by $N$.
\end{proposition}


In the following we consider the more interesting question of whether any sequential computation in an interaction-net system can be simulated on a Turing machine.
Let $(\FS,\RS)$ denote a finite interaction-net system and let $N$ denote a net based on $\FS$. The number of cells of a net $N$ is denoted as $\netsize{N}$.
The next auxiliary lemma is straightforward.
\label{def:K}

Let $K \defsym \max \{ \netsize{R} \mid { L \reducts R} \in \RS\}$.

\begin{lemma}
\label{l:1}
Suppose $N \reducts M$ for some net $M$.
Then, $\netsize{M} \leqslant \netsize{N}+K$.
\end{lemma}
\begin{proof}

Suppose that $N \reducts M$ is due to the contraction of 
the rewrite rule $L \reducts R$. Then
the number of cells in $M$ is bounded as:
$\netsize{M} = \netsize{N} - \netsize{L} + \netsize{R} \leqslant \netsize{N} + K$.
\end{proof}

We provide an implementation of interaction-net reductions on TMs.
When we say \emph{computable in time} or \emph{space} below, we implicitly mean computable on a many-tape TM $\mathcal{T} = (Q,\Sigma,\delta,\tmstart,\tmstop)$ with dedicated input and output tape.
The precise number of tapes will become clear from the construction provided.

For each $\alpha \in \FS$, we assume the existence of a tape symbol $\code{\alpha}$  in the alphabet $\Sigma$ of the encoding TM $\mathcal{T}$.
We assume that all nodes, i.e. occurrences of symbols, are consecutively numbered from $1$ to $\netsize{N}$.
We assume further that the ports of $\alpha$ are labelled from $0$ to $n$, where
$n$ is the arity of $\alpha$ and the principal port is given label $0$.

\begin{definition}
The net $N$ is encoded as an adjacency list of entries representing nodes.
We store each node $a$ in binary. In addition we store the code $\code{\alpha}$ of the symbol $\alpha$, its arity, and the list of adjacent symbols, provided in increasing order of the port labels.
More precisely, the encoding of $N$ is given by a list containing for each $a \in N$ a quadruple
\begin{equation*}
  \langle \code{a}, \code{\alpha}, n, [\code{a_0},p_{a_0},\code{a_1},p_{a_1},\dotsc,\code{a_n},p_{a_n}] \rangle
  \tkom
\end{equation*}
where $\code{a} = (a)_2$, $\alpha \in \FS$, $n$ denotes the arity of $\alpha$,
and $p_{a_i}$ denotes the incoming port in cell $a_i$ ($i=0,\dots,n$). 
\end{definition}

Given a net $N$ with $k$ cells, its encoding can be stored in size $\bigO(\ceil{\log(k)} \cdot k)$:
each entry in the adjacency list has size $\bigO(\log(k))$ and the list has at most $k$ entries.
This motivates the following definition.
The \emph{representation size} of $N$, denoted as $\rsize{N}$, is defined as $\ceil{\log(k)} \cdot k$.

\begin{definition}
A reduction rule ${\activepair{\alpha}{\beta} \reducts R} \in \RS$ is encoded as the following triple:
\begin{equation*}
  \langle \code{\alpha}, \code{\beta}, \code{R} \rangle
  \tkom
\end{equation*}
where $\code{R}$ denotes the encoding of the right-hand side of the rule $R$ as
defined above. We employ fresh nodes $π_i$, $i∈[1,m+n]$ to encode the free ports of $R$, where $n$ ($m$) denotes the arity of~$\alpha$ ($\beta$).

\end{definition}

The next lemma shows how to compute a single-step reduction on a given
net $N$. Due to the sequential computational model of nets, the concrete contracted active pair $\activepair{\alpha}{\beta}$ can be chosen non-deterministically.
We emphasise that due to the diamond property of interaction-net systems, it does not matter which active pair is chosen first.

\begin{lemma}
\label{l:2}
Let $N \reducts M$ be a reduction.
Then the encoding of $M$ is computable in time $\bigO(\rsize{N}+K)$.
\end{lemma}
\begin{proof}
First we need to locate the active pair $\activepair{\alpha}{\beta}$ in the encoding $\code{N}$ of the net $N$.
This requires at most $\rsize{N}$ steps for searching \emph{an} entry
\begin{equation*}
 \langle \code{a},\code{\alpha},n,[\code{b},p_{b},\code{b_1},p_{b_1},\dotsc,\code{b_n},p_{b_n}] \rangle
 \tkom
\end{equation*}
where the label of node $b$ is $\beta$.
We write the encoding of the reduction rule ${\activepair{\alpha}{\beta} \reducts R}$ on a worktape of the TM.
This operation is linear in $\rsize{R}$.
By definition of the encoding the free ports in $R$ are encoded as $π_i$, $i∈[1,m+n]$, where $m$ denotes the arity of $\beta$.
By construction the fresh nodes $π_i$, $i∈[1,n]$, correspond to the nodes $\code{a_i}$ in $N$, while the fresh nodes $π_{n+j}$, $j∈[1,m]$, correspond to the $\code{b_j}$.
Hence it suffices to replace the nodes $π_i$, $i∈[1,m+n]$ accordingly.
This replacement is performed directly on the encoding of $R$.
Searching takes at most time linear in $\rsize{N}$, replacing takes at most time linear in $\rsize{R}$.
Finally, we copy the remaining part of $N$ and the altered right-hand side $R$ to the output tape.
This concludes the computation of $M$.
In total, the computation takes time linear in $\rsize{N}+\rsize{R}$.
\end{proof}

\begin{lemma}
\label{l:3}
Suppose $N \reducts[\ell]_s M$ holds. Then $M$ is computable in
$\bigO((\ceil{\log(k+\ell)}) \cdot (k + \ell) \cdot \ell)$, where $k \defsym \netsize{N}$.
\end{lemma}
\begin{proof}
Consider the reduction sequence $N \reducts[\ell]_s M$.
We do not implement this sequence precisely, but allow permutations of applications of reduction rules.
Employing Lemma~\ref{l:2} the sequence $N \reducts[\ell]_s M$
becomes computable on a TM:
\begin{equation*}
  N \reducts N_1 \reducts \cdots \reducts_s N_\ell = M
  \tpkt
\end{equation*}
Let $k \defsym \netsize{N}$.
Due to Lemma~\ref{l:1}, for all $i∈[1,\ell]$, $\netsize{N_i} \leqslant \netsize{N} + i \cdot K$, where $K \in \N$.
Hence $\netsize{N_i} \in \bigO(k+\ell)$.
As $K$ only depends on the set of reduction rules $\RS$, Lemma~\ref{l:2}
yields that every single computation in this sequence requires time $\bigO(\rsize{N_i}) \subseteq \bigO(\ceil{\log(k+\ell)} \cdot (k+\ell))$.
As $\ell$ steps need to be performed, the lemma follows.
\end{proof}

\begin{corollary}
Let $(\FS,\RS)$ be a interaction-net system and let $N$ denote a net based on $\FS$ with a 
normal form $N'$, such that $N \reducts[\ell] N'$.
If we assume that $\netsize{N} \leqslant \ell$, then $N'$ is computable in time $\bigO(\log(\ell) \cdot \ell^2)$.
\end{corollary}
\begin{proof}
Let $k \defsym \netsize{N}$ and assume $k \leqslant \ell$.
The corollary is direct from the last lemma, as we obtain:
\begin{align*}
  (\ceil{\log(k+\ell)}) \cdot (k + \ell) \cdot \ell
  & \leqslant (\ceil{\log(2\ell)} \cdot 2 \ell^2
  \\
  & \leqslant 2 \cdot \ceil{\log(\ell)} \cdot \ell^2
  \\
  & \in \bigO(\ceil{\log(\ell)} \cdot \ell^2)
    \tkom
\end{align*}
from which the corollary follows.
\end{proof}

\subparagraph{Discussion}
As already mentioned, the above encoding of nets requires a logarithmic
overhead in space usage, if we opt for a unitary cost model for space that
counts only the number of cells.
In this case, the space occupation of a net $N$ equals $\netsize{N}$, 
whereas the encoding of $N$ on a TM requires size 
$\bigO(\rsize{N}) = \bigO(\ceil{\log(\netsize{N})} \cdot \netsize{N})$.
Alternatively, we propose to measure the space usage of a net according to the size of its representation.
This is sensible from a practical point of view and corresponds to the size measure for random access machines (RAMs for short) proposed in the literature~\cite{SlotB88}.

Last, we conjecture that the sequential interaction-net cost model corresponds more precisely (with linear overhead in time and space) to a RAM model with unbounded indirection and that the parallel interaction-net cost model corresponds to a PRAM model.
We will clarify these points in future work.

\fi

\section{Sized Types and Semantic Complexity}
\label{s:sized-types}

The usual notion of typing provides information concerning the expected shape of inputs and outputs.
In order to control the size of natural numbers we introduce a type $\tnat{n}$ for natural numbers whose value is bounded by $n$, thanks to the following typing schemes:
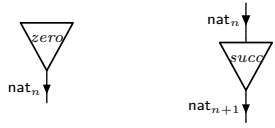
\display{
\input{asy/a583113f/\@netformat-\netstyle}
  \and
\input{asy/011ceb01/\@netformat-\netstyle}
}

Given that any object of type $\tnat{n}$ can also be given type $\tnat{m}$ if $n ≤ m$, we say that $\tnat{n}$ is a subtype of $\tnat{m}$, written $\tnat{n} ⊴ \tnat{m}$.
The conversion from $\tnat{n}$ to $\tnat{m}$ for $n ≤ m$ can be performed
using an explicit type-conversion node:
\display{
\input{asy/f2e57296/\@netformat-\netstyle}
}

We will consider its reduction rules instantaneous (i.e. they are attributed time duration $0$) because type conversions are not required for actual computation:
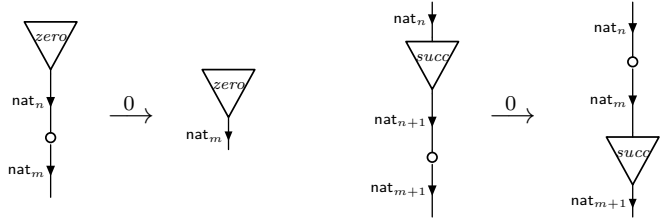
\display{
  \op{\mb{\reducts[0]}}{
\input{asy/f0832a7f/\@netformat-\netstyle}
    \and
\input{asy/5dcc3df9/\@netformat-\netstyle}
  }
  \and
  \op{\mb{\reducts[0]}}{
\input{asy/5cf65969/\@netformat-\netstyle}
    \and
\input{asy/7a0ea146/\@netformat-\netstyle}
  }
}

To control the size of lists we can similarly introduce a type $\tlist{n}{A}$ for lists whose length is bounded by $n$, thanks to the following typing schemes:
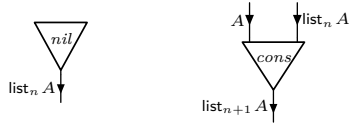
\display{
\input{asy/b18ef57a/\@netformat-\netstyle}
  \and
\input{asy/d6c96f3a/\@netformat-\netstyle}
}
We easily obtain $\tlist{n}{A} ⊴ \tlist{m}{A}$ if $n ≤ m$ using similar conversion rules.
More generally, $\tlist{n}{A} ⊴ \tlist{m}{B}$ if $n ≤ m$ and $A ⊴ B$.
The depth or the number of nodes or various particular size measures of other tree-like data structures could be tracked in the same fashion.
More elaborate data structures, e.g., difference lists from \cite{interaction-nets--y-lafont}, can be handled as well.

Subtyping is essential and used to convert a strong type constraint to a weaker constraint.
Whenever $A ⊴ B$, one is allowed to use an object of type $A$ where an object of type $B$ is expected, using type-conversion nodes of the following shape, which we allow to appear in reduction rules:
\display{
\input{asy/d40e0f46/\@netformat-\netstyle}
}

As shown below, the following operations on natural numbers and lists can be assigned the following typing schemes:
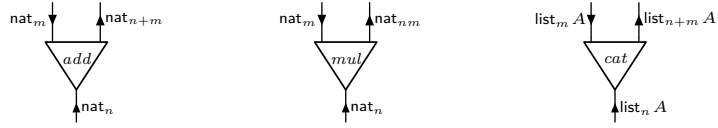
\display{
\input{asy/e97f668a/\@netformat-\netstyle}
  \and
\input{asy/a6cd49ba/\@netformat-\netstyle}
  \also
\input{asy/b5ab9a98/\@netformat-\netstyle}
}

Addition reductions satisfy typing requirements as follows (we always use the most generic typing for left-hand sides in order to handle all possible valid interactions):
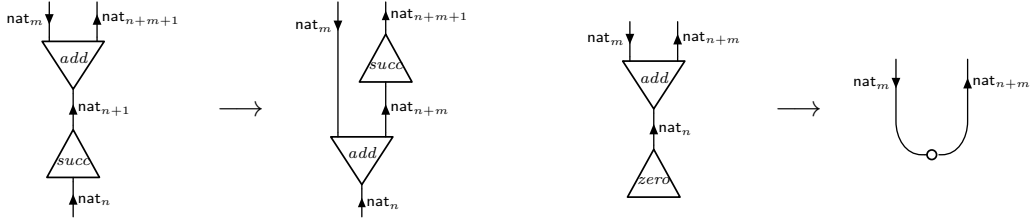
\display{
  \opreducts{
\input{asy/2e29fbf1/\@netformat-\netstyle}
    \and
\input{asy/40361d6d/\@netformat-\netstyle}
  }
  \and
  \opreducts{
\input{asy/488bb0d9/\@netformat-\netstyle}
    \and
\input{asy/45441c06/\@netformat-\netstyle}
  }
}


\ifextended
Addition and multiplication can be defined computationally in many different ways.
We have only considered one implementation of addition, which can be written as follows in a classical functional programming syntax:
\display{
  \tb{\boxcode{
      add zero y = y\\
      add (succ x) y = succ (add x y)
  }}
}

Multiplication can also be defined in many different ways.
We will consider the following three, which depend on a previously chosen implementation of addition:
\display{
  \tb{1. \boxcode{
      mul zero y = y\\
      mul (succ x) y = add y (mul x y)
  }}
  \and
  \tb{2. \boxcode{
      mul zero y = y\\
      mul (succ x) y = add (mul x y) y
  }}
  \and
  \tb{3. \boxcode{
      mul zero y = zero\\
      mul x zero = zero\\
      mul (succ x) (succ y) = succ (add y (mul x (succ y)))
  }}
}
\fi

Specialized versions of the usual interaction-net combinators $δ$ and $ε$ (see \cite{interaction-combinators--y-lafont}, they are also called “Dupl” and “Erase” in \cite{interaction-nets--y-lafont}) are used to copy or delete natural numbers.
Their reductions are the following:
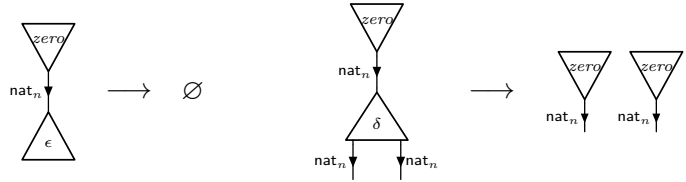
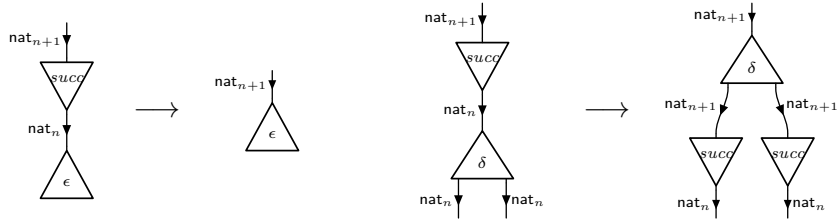
\display{
  \opreducts{
\input{asy/2d962ce1/\@netformat-\netstyle}
    \and
    \mb{\netempty}
  }
  \and
  \opreducts{
\input{asy/583946f6/\@netformat-\netstyle}
    \and
\input{asy/1df4b2f5/\@netformat-\netstyle}
  }
  \also
  \opreducts{
\input{asy/b136ca85/\@netformat-\netstyle}
    \and
\input{asy/f17b394a/\@netformat-\netstyle}
  }
  \and
  \opreducts{
\input{asy/257e413f/\@netformat-\netstyle}
    \and
\input{asy/ed3d859e/\@netformat-\netstyle}
  }
}

The most standard multiplication (variant 1) is defined below with interaction nets.
The adequate typing of its reduction rules validates the size complexity property expressed in its typing scheme.
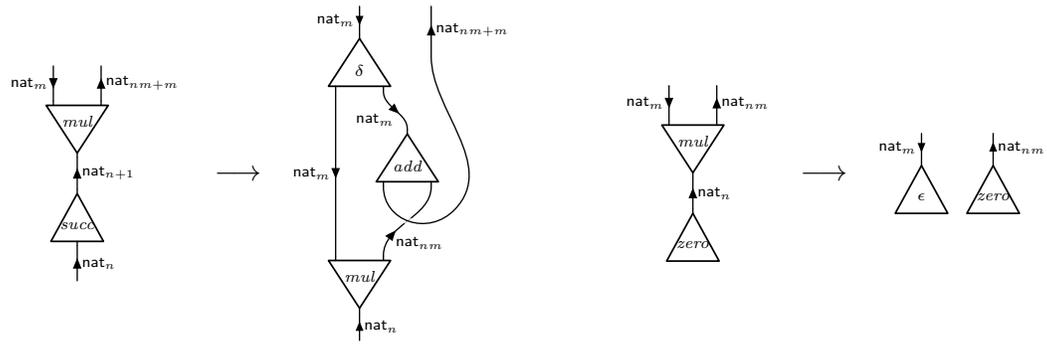
\display{
  \opreducts{
\input{asy/acfc92f8/\@netformat-\netstyle}
    \and
\input{asy/43286ed9/\@netformat-\netstyle}
  }
  \and
  \opreducts{
\input{asy/2b1130ea/\@netformat-\netstyle}
    \and
\input{asy/f31328bb/\@netformat-\netstyle}
  }
}

Multiplication (variant 2) is the one presented in \cite{interaction-nets--y-lafont}.
It has different computational properties (it is generally less efficient), but it can be typed similarly:
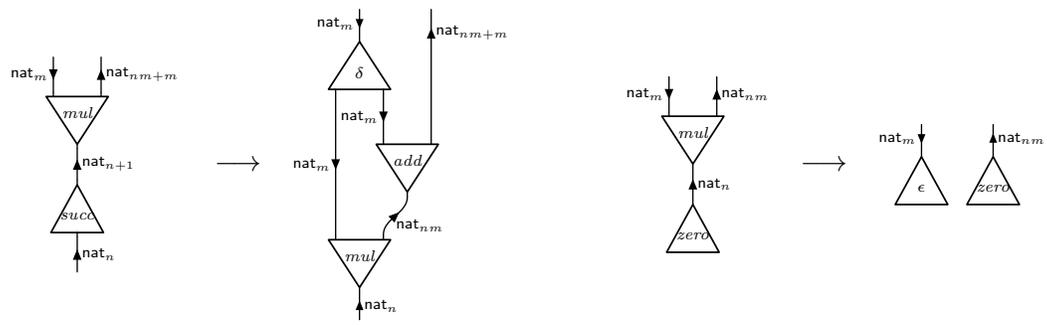
\display{
  \opreducts{
\input{asy/acfc92f8/\@netformat-\netstyle}
    \and
\input{asy/99ca2198/\@netformat-\netstyle}
  }
  \and
  \opreducts{
\input{asy/2b1130ea/\@netformat-\netstyle}
    \and
\input{asy/f31328bb/\@netformat-\netstyle}
  }
}

\ifextended
Multiplication (variant 3) is, as we will show later, a better fit for parallel computation and can be defined using an auxiliary node as follows:
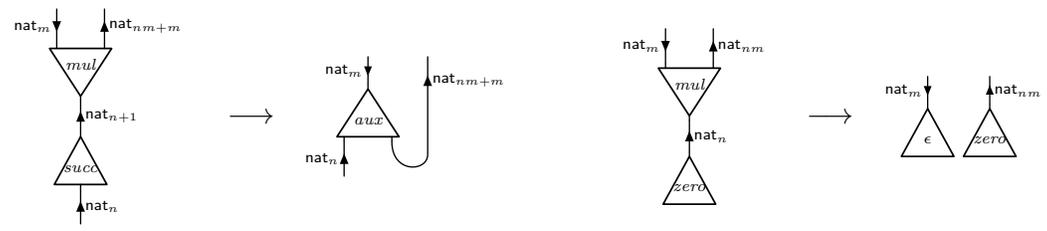
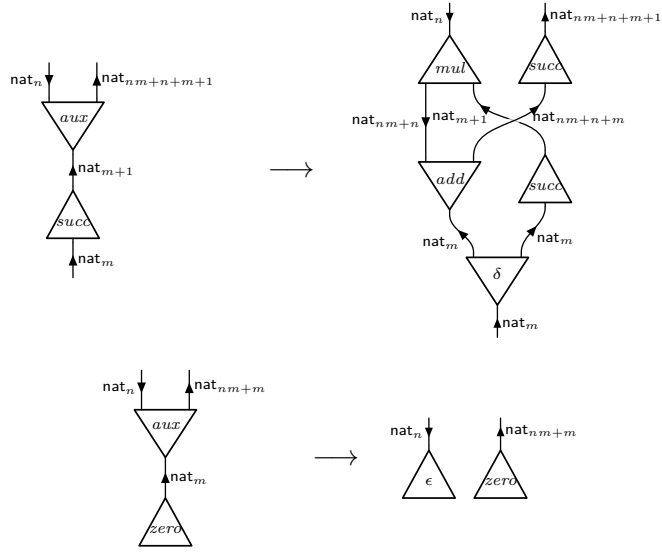
\display{
  \opreducts{
\input{asy/acfc92f8/\@netformat-\netstyle}
    \ifextended\else\hskip -15pt\fi
    \and
\input{asy/a6556a41/\@netformat-\netstyle}
  }
  \and
  \opreducts{
\input{asy/1812b675/\@netformat-\netstyle}
    \ifextended\else\hskip -15pt\fi
    \and
\input{asy/d8f33eda/\@netformat-\netstyle}
  }
  \also
  \opreducts{
\input{asy/9b590c46/\@netformat-\netstyle}
    \ifextended\else\hskip -30pt\fi
    \and
\input{asy/4980894b/\@netformat-\netstyle}
  }
  \and
  \opreducts{
\input{asy/0af214ea/\@netformat-\netstyle}
    \ifextended\else\hskip -15pt\fi
    \and
\input{asy/1fa8d3cf/\@netformat-\netstyle}
  }
}
\fi

All implementations possess their own particular computational complexity properties (each could be more efficient in a given context), but all additions (respectively, all multiplications) are semantically equivalent and their outputs share the same size bound property, as expressed by their common typing scheme.

Size bounds for list concatenation are obtained as follows:
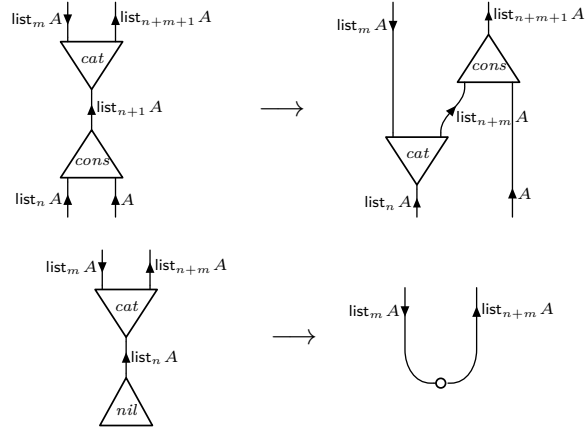
\display{
  \opreducts{
\input{asy/74b6c70e/\@netformat-\netstyle}
    \and
\input{asy/33ce669c/\@netformat-\netstyle}
  }
  \and
  \opreducts{
\input{asy/13f9dcb5/\@netformat-\netstyle}
    \and
\input{asy/e340b155/\@netformat-\netstyle}
  }
}

\section{Sequential Computational Complexity Analysis}
\label{s:sequential-analysis}

We will prove a space–time complexity theorem for the sequential reduction, which will then be turned into separate space complexity and time complexity theorems.
However, we provide first a simple example explaining how our analysis allows one to infer the time complexity of multiplication (variant 2) from those of the operations used in its definition.
\begin{example}
Assuming knowledge of time potentials $τ$ (which, as will be explained soon, correspond closely to time complexities) for the following typed nodes,
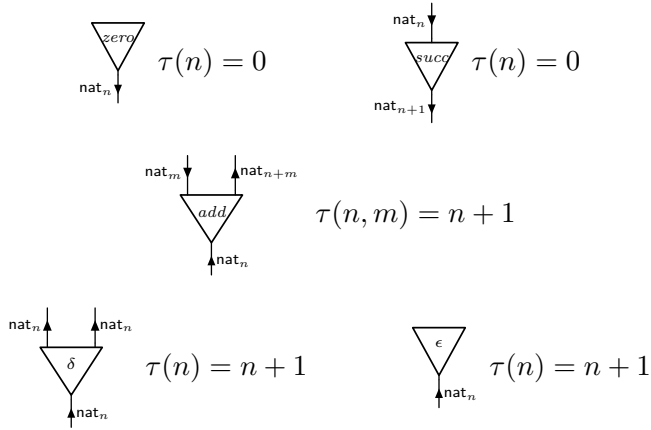
\display{
  \cscheme{\sc{τ}(n) = 0}{
\input{asy/99f546d2/\@netformat-\netstyle}
  }
  \and
  \cscheme{\sc{τ}(n) = 0}{
\input{asy/72f7ebe2/\@netformat-\netstyle}
  }
  \also
  \cscheme{\sc{τ}(n, m) = n + 1}{
\input{asy/2f659fec/\@netformat-\netstyle}
  }
  \also
  \cscheme{\sc{τ}(n) = n + 1}{
\input{asy/221dc934/\@netformat-\netstyle}
  }
  \and
   \cscheme{\sc{τ}(n) = n + 1}{
\input{asy/cbc9d24e/\@netformat-\netstyle}
  }
}
\noindent our results ensure that multiplication (variant 2) as defined previously can be attributed potential (and complexity):
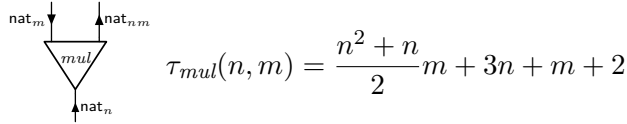
\display{
  \cscheme{\sc{τ}_{\ntext{mul}}(n, m)=\dfrac{n^2 + n}{2}m + 3n + m + 2}{
\input{asy/e2e663f6/\@netformat-\netstyle}
  }
}

The only requirement is to check that, in all reduction rules, the sum of the time potentials assigned to typed nodes in the left-hand side is strictly greater than the sum of the time potentials assigned to typed nodes in the right-hand side.
In this example, the constraints associated to multiplication reduction rules, $\sc{τ}_{\ntext{mul}}(n, m) + 0 > (m+1) + 0$ and $\sc{τ}_{\ntext{mul}}(n+1, m) + 0 > (nm+1) + (m+1) + \sc{τ}_{\ntext{mul}}(n, m)$, can be checked or solved easily to obtain the result.
\end{example}

We introduce notations which we will use later to represent concrete space–time complexities. Given a function $f : T \rightarrow S$, where $T$ denotes a time domain and $S$ a space domain, we define delayed functions $[f]^d$ as $[f]^d(t) = f(t-d)$, or $[f]^d(t) = 0$ when $t < d$.
Iterated sequences are defined recursively as $f ⊐^d_0 g = f$ and $f ⊐^d_{n+1} g = (f ⊐^d_n g) + [g]^{(n+1)d}$.
Intuitively $f ⊐^d_n g$ represents the superimposition of $n$ copies of $g$, successively delayed by $d$, to $f$.
Finally, compound iterated sequences recursively extend this definition with successive superimpositions as
$f ⊐^d_n g ⊐^{d_1}_{n_1} g_1 ⊐^{d_2}_{n_2} \dotsb ⊐^{d_p}_{n_p} g_p = (f ⊐^d_n g) ⊐^{d_1}_{n_1} [g_1]^{nd} ⊐^{d_2}_{n_2} \dotsb ⊐^{d_p}_{n_p} [g_p]^{nd}$.
The exponent or index of any $⊐$-symbol defaults to $1$ if omitted.
We abusively use numerals to denote constant functions.

For example, here is a graphical representation of $γ = 1 ⊐_3 1 ⊐_4 0 ⊐_2 -2$.
It initially has value $1$ and successively undergoes (after unitary delays) $3$ increments of $1$, then remains constant $4$ times, and undergoes $2$ decrements of $2$.
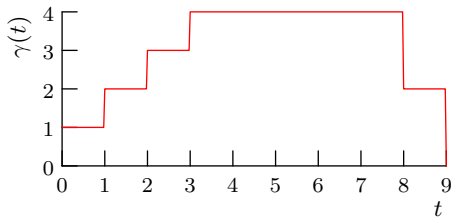
\display{
\input{asy/4150da65/\@netformat-\netstyle}
}


We define the \emph{sequential convolution} of functions $f$ and $g$ as follows:
\display{
  \mb{\abracket{f*g}(t) = \max_{u+v=t} f(u) + g(v)}
}

This operation is commutative and associative.
The generalized sequential convolution of a family of functions $\set{f_i}_{i∈I}$ can be expressed as:
\display{
  \mb{
    \abracket{\coprod_{i∈I} f_i}(t) = \max_{\sum_{i∈I} t_i = t} ~ \sum_{i∈I} f_i(t_i)
  }
}

We will range over all typed nodes of a net $N$ by indexing an operation (e.g. a sum or a sequential convolution) with $c ∈ N$.
Nodes occurring multiple times in $N$ are counted over multiple times.
In particular, assuming that every node $c$ has been assigned a space-occupation weight $|c| ≥ 0$ (weights can be chosen arbitrarily), we define the space occupation of a net as the sum of its nodes' weights: $|N| = \sum_{c∈N} |c|$.

\begin{theorem}[Sequential Space–Time Complexity]
\label{t:sequential}
Associate a function $\sc{γ}_c : T \rightarrow S$, called space–time potential, to every typed node $c$, such that $\sc{γ}_c(0) ≥ |c|$ and $\abracket{\coprod_{c∈L} \sc{γ}_c}(t+d) ≥ \abracket{\coprod_{c∈R} \sc{γ}_c}(t)$ for every reduction rule $L \reducts[d] R$.
Whenever $N \reducts[t]_s M$: $$|M| ≤ \abracket{\coprod_{c∈N} \sc{γ}_c}(t)$$
\end{theorem}
\begin{proof}
The potential-decrease property assumed for reduction rules entails the same property for individual sequential reduction steps of a net: $\sc{γ}_L(t+d) ≥ \sc{γ}_R(t) \implies (γ*\sc{γ}_L)(t+d) = \max_{u+v=t+d} (γ(u) + \sc{γ}_L(v)) ≥ \max_{u+v=t} (γ(u) + \sc{γ}_L(v+d)) ≥ \max_{u+v=t} (γ(u) + \sc{γ}_R(v)) = (γ*\sc{γ}_R)(t)$.
By combining these steps, we obtain $\abracket{\coprod_{c∈N} \sc{γ}_c}(t) ≥ \dotsb ≥ \abracket{\coprod_{c∈M} \sc{γ}_c}(0) = \sum_{c∈M} \sc{γ}_c(0) ≥ |M|$.
\end{proof}

\begin{corollary}[Sequential Time Complexity]
\label{c:sequential-time}
Associate $\sc{τ}_c ≥ 0$, called time potential, to every typed node $c$, such that $\sum_{c∈L} \sc{τ}_c ≥ \sum_{c∈R} \sc{τ}_c + d$ for every reduction rule $L \reducts[d] R$.
Whenever $N \reducts[t]_s M$: $$t ≤ \sum_{c∈N} \sc{τ}_c$$
\end{corollary}
\begin{proof}
Choose $\sc{γ}_c(t) = \sc{τ}_c - t$ and null space-occupation weights; remark that $\abracket{\coprod_{c∈N} \sc{γ}_c}(t) = \sum_{c∈N} \sc{τ}_c - t$.
\end{proof}

\begin{corollary}[Sequential Space Complexity]
\label{c:sequential-space}
Associate $\sc{σ}_c ≥ |c|$, called space potential, to every typed node $c$, such that $\sum_{c∈L} \sc{σ}_c ≥ \sum_{c∈R} \sc{σ}_c$ for every reduction rule $L \reducts[d] R$.
Whenever $N \reducts[t]_s M$: $$|M| ≤ \sum_{c∈N} \sc{σ}_c$$
\end{corollary}
\begin{proof}
Choose constant functions $\sc{γ}_c(t) = \sc{σ}_c$; remark that $\abracket{\coprod_{c∈N} \sc{γ}_c}(t) = \sum_{c∈N} \sc{σ}_c$.
\end{proof}

Subtyping conversion rules such as those we have defined satisfy the sequential potential-decrease property; because they have been assigned durations $0$, it is sufficient to ensure that the potential assigned to constructors is monotonic with respect to subtyping.
For lack of precise knowledge about the hardware that may host the computation, we choose to assign a unitary space-occupation weight to all nodes.
For our present needs, constructors will be assigned constant witness potentials, but they could be assigned different potentials if we were to perform amortized cost analysis in full.
The following potentials are associated to the displayed typing schemes and parameterized in their size variables.
They are compatible with all reduction rules and therefore provide sequential time ($τ$), space ($σ$) and space–time ($γ$) complexity measures:
\display{
  \cscheme{\sc{τ} = 0 \cr \sc{σ} = 1 \cr \sc{γ} = 1}{
\input{asy/99f546d2/\@netformat-\netstyle}
  }
  \and
  \cscheme{\sc{τ} = 0 \cr \sc{σ} = 1 \cr \sc{γ} = 1}{
\input{asy/72f7ebe2/\@netformat-\netstyle}
  }
  \also
  \cscheme{\sc{τ} = n + 1 \cr \sc{σ} = 1 \cr \sc{γ} = 1 ⊐_n 0 ⊐ -2}{
\input{asy/2f659fec/\@netformat-\netstyle}
  }
  \and
  \cscheme{\sc{τ} = 2nm + 2n + m + 2 \cr \sc{σ} = nm + n + 1 \cr
    \sc{γ} = 1 ⊐_{nm+n} 1 ⊐_{nm+1} 0 ⊐_m -1 ⊐_{n+2} -2}{
    \text{(var. 1) \hskip -15pt}
\input{asy/e2e663f6/\@netformat-\netstyle}
  }
  \also
  \cscheme{\sc{τ} = n + 1 \cr \sc{σ} = n + 1 \cr \sc{γ} = 1 ⊐_n 1 ⊐ 0}{
\input{asy/221dc934/\@netformat-\netstyle}
  }
  \and
   \cscheme{\sc{τ} = n + 1 \cr \sc{σ} = 1 \cr \sc{γ} = 1 ⊐_n -1 ⊐ -2}{
\input{asy/cbc9d24e/\@netformat-\netstyle}
  }
}

\subparagraph{Example of a concrete application to program complexity certification}
If we consider the net built by providing interaction-net representations of two natural numbers bounded respectively by $n$ and $m$ as arguments to a $\ntext{mul}$ node, the combined time and space potentials of all the nodes are bounded respectively by $τ = 2nm + 2n + m + 2$ and $σ = nm + 2n + m + 3$ (constructors in the representation of natural number $i$ possess themselves a combined potential of $τ = 0$ and $σ = i + 1$).
According to Corollaries~\ref{c:sequential-time} and~\ref{c:sequential-space}, these combined potentials constitute concrete bounds on the sequential time and space needed to perform the multiplication as a function of the size of its arguments.
Furthermore, the bounds are precise.

\section{Scheduled Types and Productivity}
\label{s:scheduled-types}

For the care of fully parallel reductions, we distinguish several categories of natural numbers depending on the pace at which they are computed.
A natural number admits type $\tnatp{d}$ if its first constructor is available now and one additional constructor will be made available after every parallel reduction of duration $d ≥ 0$ (or faster).

This corresponds to the following typing scheme definitions, in which we use a bracket notation to denote \emph{delayed types}: any object of type $[A]^t$ will become an $A$ after a parallel reduction of duration $t ≥ 0$ (defaulting to $1$ if omitted).
We assume syntactic equalities $A = [A]^0$ and $[[A]^{t_1}]^{t_2} = [A]^{t_1+t_2}$.
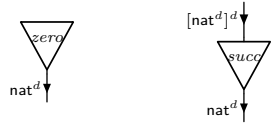
\display{
\input{asy/c760ef0e/\@netformat-\netstyle}
  \and
\input{asy/afd4cd31/\@netformat-\netstyle}
}

Similarly we can define a type $\tlistp{d}{A}$ for linearly produced lists with pace $d$ as follows:
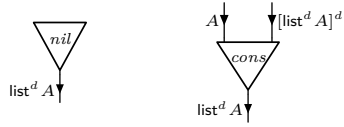
\display{
\input{asy/61e0bb5b/\@netformat-\netstyle}
  \and
\input{asy/0b897756/\@netformat-\netstyle}
}

Restricting ourselves to unitary time reductions, we will show that implementations of addition, multiplication \ifextended (variant 3)\fi and list concatenation exist with the following interfaces:
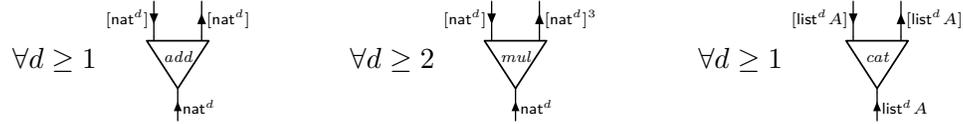
\display{
  \mb{∀d ≥ 1
\input{asy/74994e08/\@netformat-\netstyle}
  }
  \and
  \mb{∀d ≥ 2
\input{asy/c16994b8/\@netformat-\netstyle}
  }
  \also
  \mb{∀d ≥ 1
\input{asy/313bee8d/\@netformat-\netstyle}
  }
}

It is assumed that:
\begin{itemize}
\item In typing schemes and left-hand sides of reduction rules, principal ports are assigned undelayed types (i.e. the root is a type variable or a base type).
\item The initial delays present in the outputs of a reduction rule are reduced by the rule's duration upon firing.
Input delays may but need not be reduced by the full amount of the reduction rule's duration.
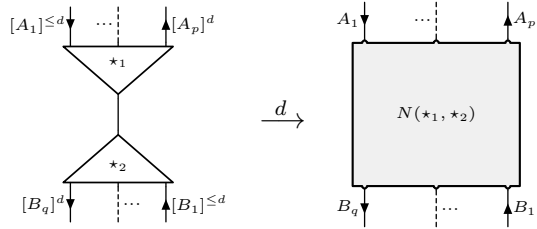
\display{
  \op{\mb{\reducts[d]}}{
\input{asy/5e1c96fc/\@netformat-\netstyle}
    \and
\input{asy/70c764a1/\@netformat-\netstyle}
  }
}
\end{itemize}

In particular, it is always assumed that top-level output delays in the interface of a redex are greater than the assigned duration of the corresponding reduction rule.

We allow two typing conveniences for scheduled types:
\begin{itemize}
\item Subtyping: Assuming $A$ is the type associated to a tree-like data structure (i.e. constructors only have inputs, which includes $\mathsf{nat}$ and $\mathsf{list}$ types, but excludes for example difference lists as in \cite{interaction-nets--y-lafont} or abstractions), we have $A ⊴ [A]^t$.
One can safely use an object of type $A$ where an object of type $[A]^t$ is expected.
In particular, we deduce $\tnatp{d_1} ⊴ \tnatp{d_2}$ if $d_1 ≤ d_2$, and $\tlistp{d_1}{A} ⊴ \tlistp{d_2}{B}$ if $d_1 ≤ d_2$ and $A ⊴ B$.

\item Delayed computation: The execution of a net can be delayed by delaying all the types in its interface (independently of their orientations).
For a single node, this allows to instantiate any typing scheme $c$ as $[c]^d$:
\display{
  \mb{c =
\input{asy/37dbb48c/\@netformat-\netstyle}
  }
  \and
  \mb{[c]^d =
\input{asy/fe83ab6b/\@netformat-\netstyle}
  }
}
\end{itemize}


Within those conditions, the above property assumed for reduction rules can be generalized to the fully parallel reduction of nets.
If we assume that inputs are available within expected schedules, the parallel reduction of a net will produce outputs in accordance with the schedules that correspond to their types.
For example, an output wire typed $\tnatp{d}$ will output natural number constructors with pace $d$.

Reduction rules for addition of natural numbers are as follows ($d ≥ 1$):
\display{
  \opreducts{
\input{asy/dd3cbc91/\@netformat-\netstyle}
    \and
\input{asy/e141ea1c/\@netformat-\netstyle}
  }
  \and
  \opreducts{
\input{asy/733fa21d/\@netformat-\netstyle}
    \and
\input{asy/f36f8561/\@netformat-\netstyle}
  }
}

The typing and reduction rules of the specialized versions of $δ$ and $ε$ to natural numbers are straightforward ($d ≥ 1$):
\display{
  \opreducts{
\input{asy/02f57303/\@netformat-\netstyle}
    \and
    \mb{\netempty}
  }
  \and
  \opreducts{
\input{asy/ebfbac58/\@netformat-\netstyle}
    \and
\input{asy/580d6aac/\@netformat-\netstyle}
  }
  \also
  \opreducts{
\input{asy/da9e1deb/\@netformat-\netstyle}
    \and
\input{asy/4b460fa6/\@netformat-\netstyle}
  }
  \and
  \opreducts{
\input{asy/559b0db5/\@netformat-\netstyle}
    \and
\input{asy/f24defcf/\@netformat-\netstyle}
  }
}

For multiplication (variant 1) of natural numbers, the timing analysis has to rely on the size analysis.
The output is produced after a significant delay that depends (linearly) on the size of the first parameter.
This example illustrates the combined use of size and timing annotations which is generally required to obtain precise complexity bounds.
For $d ≥ 1$, we obtain:
\display{
\input{asy/a8929ae3/\@netformat-\netstyle}
}

Reduction rules are as follows:
\display{
  \opreducts{
\input{asy/243ed145/\@netformat-\netstyle}
    \and
\input{asy/70eecdc0/\@netformat-\netstyle}
  }
  \and
  \opreducts{
\input{asy/349eb454/\@netformat-\netstyle}
    \and
\input{asy/f2ecbac4/\@netformat-\netstyle}
  }
}

Multiplication (variant 3) of natural numbers is defined using an auxiliary node and constitutes in most cases an interesting improvement in a parallel model as it offers a constant delay.
For $d ≥ 2$, we obtain:
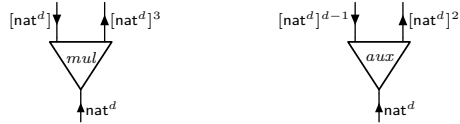
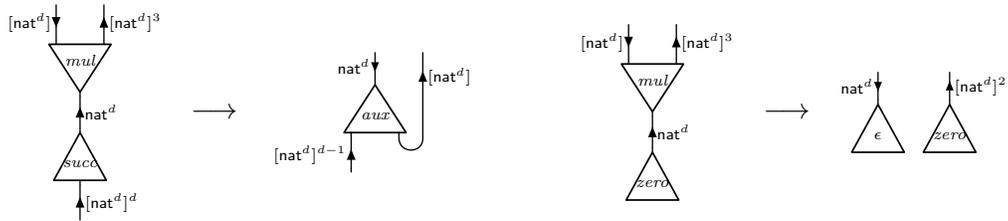
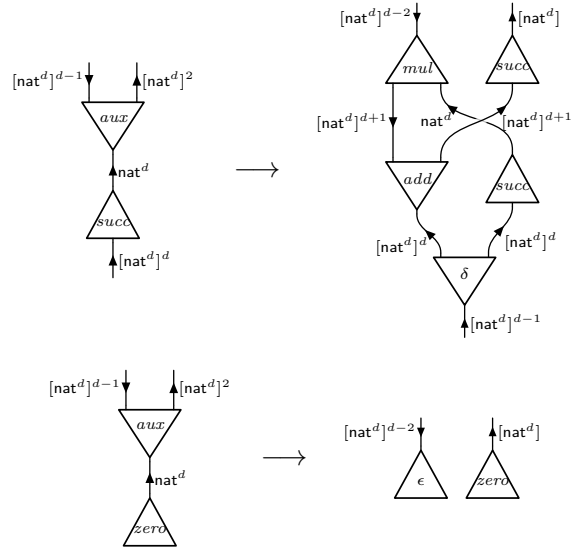
\display{
\input{asy/c16994b8/\@netformat-\netstyle}
  \and
\input{asy/4635b053/\@netformat-\netstyle}
}

Reduction rules:
\display{
  \opreducts{
\input{asy/f0ccc19f/\@netformat-\netstyle}
    \and
\input{asy/7002eca5/\@netformat-\netstyle}
  }
  \and
  \opreducts{
\input{asy/b434dc1d/\@netformat-\netstyle}
    \and
\input{asy/0fb632a0/\@netformat-\netstyle}
  }
  \also
  \opreducts{
\input{asy/296cedda/\@netformat-\netstyle}
    \and
\input{asy/123a7120/\@netformat-\netstyle}
  }
  \and
  \opreducts{
\input{asy/022fa502/\@netformat-\netstyle}
    \and
\input{asy/1de11e11/\@netformat-\netstyle}
  }
}

Reduction rules for list concatenation are as follows:
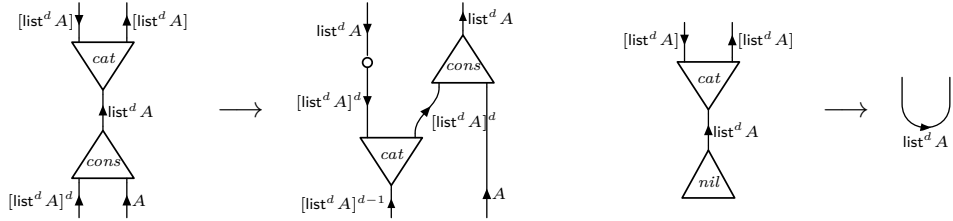
\display{
  \opreducts{
\input{asy/39512f8e/\@netformat-\netstyle}
    \and
\input{asy/da09ae81/\@netformat-\netstyle}
  }
  \and
  \opreducts{
\input{asy/b15b36b8/\@netformat-\netstyle}
    \and
\input{asy/86152a81/\@netformat-\netstyle}
  }
}

\section{Parallel Computational Complexity Analysis}
\label{s:parallel-analysis}

We prove a space–time complexity theorem for the fully parallel reduction, which is then turned into a separate time complexity theorem.
We rely on schedule-typed nodes, i.e. nodes with scheduled typing schemes that include delay connectives and follow the schedule requirements that have been defined in the previous section.

\begin{theorem}[Parallel Space–Time Complexity]
\label{t:parallel}
Associate a function $\pc{γ}_c : T \rightarrow S$, called parallel space–time potential, to every schedule-typed node $c$, such that $\pc{γ}_c(0) ≥ |c|$, $\pc{γ}_{[c]^d}(t+d) ≥ \pc{γ}_c(t)$ and $\abracket{\sum_{c∈L} \pc{γ}_c}(t+d) ≥ \abracket{\sum_{c∈R} \pc{γ}_c}(t)$ for every reduction rule $L \reducts[d] R$.
Whenever $N \reducts[t]_p M$: $$|M| ≤ \abracket{\sum_{c∈N} \pc{γ}_c}(t)$$
\end{theorem}
\begin{proof}
The potential-decrease property assumed for reduction rules entails the same property for atomic parallel reduction steps of a net.
Hence, by combining these steps, we obtain $\abracket{\sum_{c∈N} \pc{γ}_c}(t) ≥ \dotsb ≥ \abracket{\sum_{c∈M} \pc{γ}_c}(0) = \sum_{c∈M} \pc{γ}_c(0) ≥ |M|$.
\end{proof}

\begin{corollary}[Parallel Time Complexity]
Associate $\pc{τ}_c ≥ 0$, called parallel time potential, to every schedule-typed node $c$, such that $\pc{τ}_{[c]^d} ≥ \pc{τ}_c + d$ and $\max_{c∈L} \pc{τ}_c ≥ \max_{c∈R} \pc{τ}_c + d$ for every reduction rule $L \reducts[d] R$.
Whenever $N \reducts[t]_p M$: $$t ≤ \max_{c∈N} \pc{τ}_c$$
\end{corollary}
\begin{proof}
Choose $\pc{γ}_c(t) = \pc{τ}_c - t$ and null space-occupation weights; remark that $\abracket{\sum_{c∈N} \pc{γ}_c}(t) ≥ 0$ implies $t ≤ \max_{c∈N} \pc{τ}_c$.
\end{proof}


With the same assumptions as for the sequential reduction,
the following potentials are compatible with reduction rules and therefore provide parallel time ($\pc{τ}$) and space–time ($\pc{γ}$) complexity measures.
Notice that this analysis reveals that multiplication (variant 3), which has an output of quadratic size, can be performed in linear time with a parallel reduction.
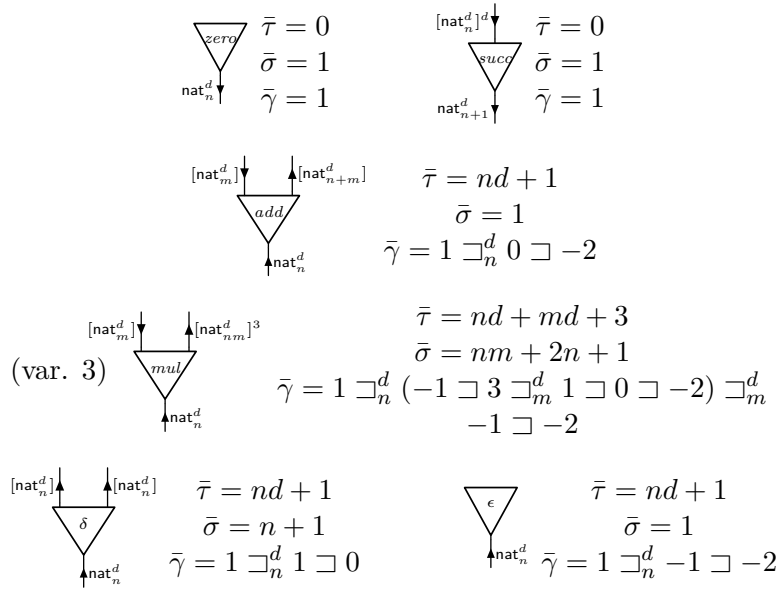
\todo{fix values}
\display{
  \cscheme{\pc{τ} = 0 \cr \pc{σ} = 1 \cr \pc{γ} = 1}{
\input{asy/be69cad3/\@netformat-\netstyle}
  }
  \and
  \cscheme{\pc{τ} = 0 \cr \pc{σ} = 1 \cr \pc{γ} = 1}{
\input{asy/1265cc12/\@netformat-\netstyle}
  }
  \also
  \cscheme{\pc{τ} = nd + 1 \cr \pc{σ} = 1 \cr \pc{γ} = 1 ⊐^d_n 0 ⊐ -2}{
\input{asy/58dc184c/\@netformat-\netstyle}
  }
  \and
  \cscheme{\pc{τ} = nd + md + 3 \cr \pc{σ} = nm + 2n + 1 \cr \pc{γ} = 1 ⊐^d_n (-1 ⊐ 3 ⊐^d_m 1 ⊐ 0 ⊐ -2) ⊐^d_m \cr -1 ⊐ -2}{
    \text{(var. 3) \hskip -15pt}
\input{asy/e5bc75e4/\@netformat-\netstyle}
  }
  \also
  \cscheme{\pc{τ} = nd + 1 \cr \pc{σ} = n + 1 \cr \pc{γ} = 1 ⊐^d_n 1 ⊐ 0}{
\input{asy/6887e5af/\@netformat-\netstyle}
  }
  \and
  \cscheme{\pc{τ} = nd + 1 \cr \pc{σ} = 1 \cr \pc{γ} = 1 ⊐^d_n -1 ⊐ -2}{
\input{asy/e69e6496/\@netformat-\netstyle}
  }
}

For parallel reductions, simple space complexities do compose, but fail to give precise bounds (this is why we omitted the corresponding corollary).
The accurate space complexities $\pc{σ}$ reported above were extracted from the complete space–time complexity analysis as the maximum value of $\pc{γ}$.

\todo{example of concrete application}

\section{Case Study: Higher Order}
\label{s:higher-order}

We can analyze higher-order functional programs in a weak sequential cost model, using the interaction-net framework extended by \emph{functorial promotion boxes} and the associated \emph{weakening}, \emph{contraction}, \emph{dereliction} and \emph{digging} nodes as presented in Section~\ref{s:interaction-nets}.
The sized versions of exponential types are simply expressed as multisets of types.
This allows the use of resources to be heterogeneous.
We use $\oc_{i<n} A_i$ as syntactic sugar to denote the multiset $\{A_0, \dotsc, A_{n-1}\}$ and write $\oc_n A$ for homogeneous multisets that contain a single element with multiplicity $n$.
The empty multiset is denoted by $\msetempty$ and the union by $\msetunion$.
Space occupation of boxes and the duration of box reductions are assumed to be unitary (which is arguably an important simplification, but one similar to attributing a constant cost to a $\beta$-reduction).

Sized typing schemes and potentials can be assigned as follows.
In particular, typing the interface of a box with multisets of size $n$ (they must have the same size) requires typing its contents $n$ times.
Requested types for the interface of the contents have been indexed by $i∈[0,n-1]$.
Each typing of the contents corresponds recursively to some resource usage $\sc{γ}_i$ (respectively $\sc{τ}_i$ or $\sc{σ}_i$).
The potential $\sc{γ}$ (respectively $\sc{τ}$ or $\sc{σ}$) of the box is expressed as a function of these resource usages.
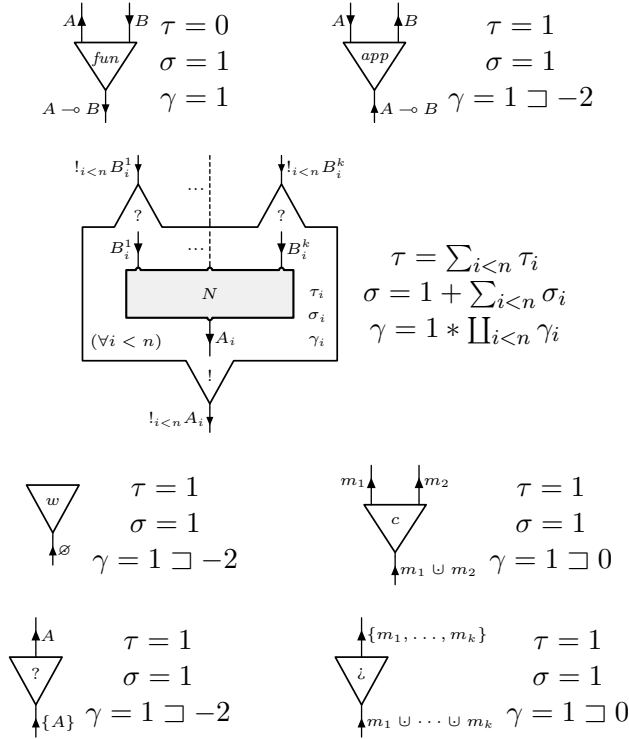
\display{
  \cscheme{\sc{τ} = 0 \cr \sc{σ} = 1 \cr \sc{γ} = 1}{
\input{asy/7545785a/\@netformat-\netstyle}
  }
  \and
  \cscheme{\sc{τ} = 1 \cr \sc{σ} = 1 \cr \sc{γ} = 1 ⊐ -2}{
\input{asy/25791321/\@netformat-\netstyle}
  }
  \also
  \cscheme{\sc{τ} = \sum_{i<n} \sc{τ}_i \cr \sc{σ} = 1 + \sum_{i<n} \sc{σ}_i \cr \sc{γ} = 1 * \coprod_{i<n} \sc{γ}_i}{
\input{asy/6502622e/\@netformat-\netstyle}
  }
  \also
  \cscheme{\sc{τ} = 1 \cr \sc{σ} = 1 \cr \sc{γ} = 1 ⊐ -2}{
\input{asy/60f4b4a8/\@netformat-\netstyle}
  }
  \and
  \cscheme{\sc{τ} = 1 \cr \sc{σ} = 1 \cr \sc{γ} = 1 ⊐ 0}{
\input{asy/453729d7/\@netformat-\netstyle}
  }
  \also
  \cscheme{\sc{τ} = 1 \cr \sc{σ} = 1 \cr \sc{γ} = 1 ⊐ -2}{
\input{asy/a2d1af1e/\@netformat-\netstyle}
  }
  \and
  \cscheme{\sc{τ} = 1 \cr \sc{σ} = 1 \cr \sc{γ} = 1 ⊐ 0}{
\input{asy/7406a85d/\@netformat-\netstyle}
  }
}

Multiset inclusion is admissible as subtyping.
One can check that with these typing schemes and potentials the related (already presented) six reduction rules satisfy the potential-decrease property.

One can also verify that the usual $\ntext{map}$ and $\ntext{fold}$ admit the following typing schemes and potentials:
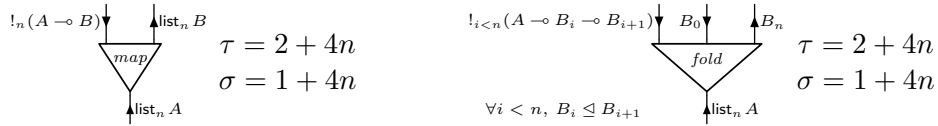
\display{
  \cscheme{\sc{τ} = 2 + 4n \cr \sc{σ} = 1 + 4n}{
\input{asy/6def2da0/\@netformat-\netstyle}
  }
  \and
  \cscheme{\sc{τ} = 2 + 4n \cr \sc{σ} = 1 + 4n}{
\input{asy/4030555e/\@netformat-\netstyle}
  }
}

\ifextended
Size-typed reduction rules for $\ntext{map}$ are as follows:
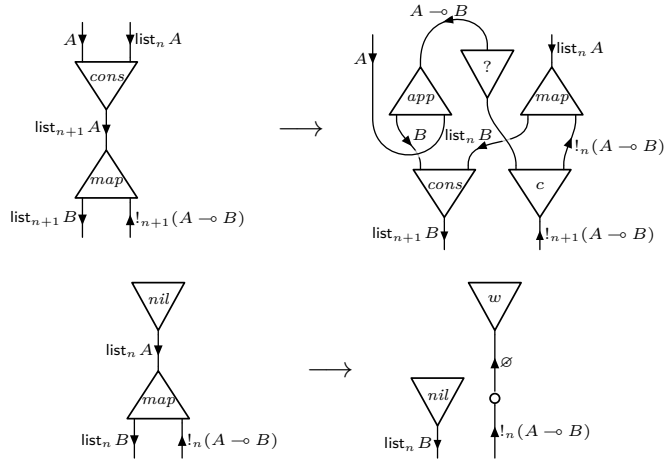
\display{
  \opreducts{
\input{asy/db17117a/\@netformat-\netstyle}
    \and
\input{asy/8a5cb261/\@netformat-\netstyle}
  }
  \and
  \opreducts{
\input{asy/4862b46b/\@netformat-\netstyle}
    \and
\input{asy/89c2efec/\@netformat-\netstyle}
  }
}

Size-typed reduction rules for $\ntext{fold}$ are as follows:
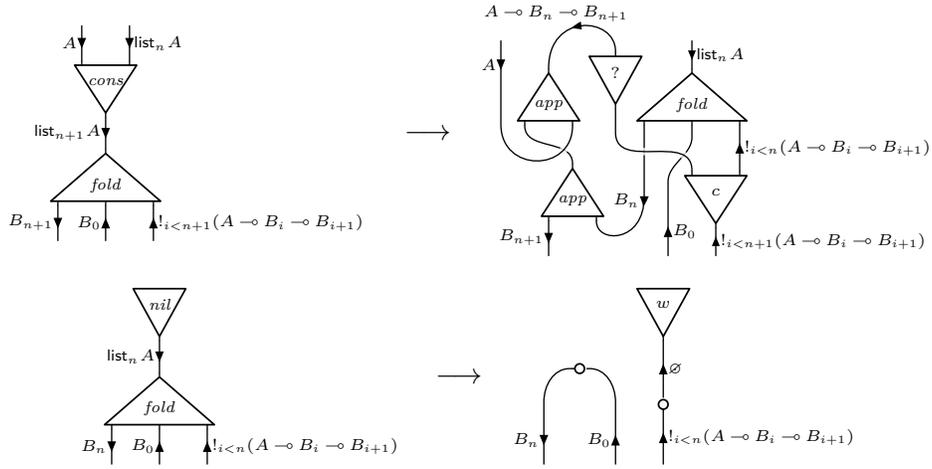
\display{
    \opreducts{
\input{asy/793ab4b5/\@netformat-\netstyle}
    \and
\input{asy/fef7a5e4/\@netformat-\netstyle}
  }
  \and
  \opreducts{
\input{asy/12f498c6/\@netformat-\netstyle}
    \and
\input{asy/604f34d3/\@netformat-\netstyle}
  }
}
\fi

Resource usage of functional arguments to $\ntext{map}$ and $\ntext{fold}$ are taken into account as the potential of the box which holds them.
For example:
\begin{itemize}
  \item Mapping a $\ntext{succ}: \tnat{p} ⊸ \tnat{p+1}$ operation with potential $\sc{τ} = 0$ and $\sc{σ} = 1$ is done as follows:
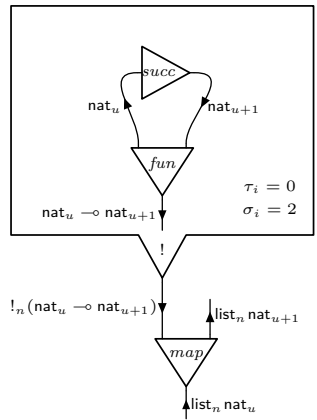
\display{
\input{asy/97978a07/\@netformat-\netstyle}
}

This defines an operation of type $\tlist{n}{\tnat{u}} ⊸ \tlist{n}{\tnat{u+1}}$ with a total potential of $\sc{τ} = 2 + 4n + \sum_{i<n} 0 = 2 + 4n$ and $\sc{σ} = 1 + 4n + 1 + \sum_{i<n} 2 = 2 + 6n$.

\item Folding an $\ntext{add} : \tnat{p} ⊸ \tnat{q} ⊸ \tnat{p+q}$ operation with potentials $\sc{τ} = 1 + p$ and $\sc{σ} = 1$ demands instantiating $A$ as $\tnat{u}$ and $B_i$ as $\tnat{v+iu}$ and defines an operation of type $\tlist{n}{\tnat{u}} ⊸ \tnat{v} ⊸ \tnat{v+nu}$ with potentials $\sc{τ} = 2 + 4n + \sum_{i<n} (1 + u) = 2 + 5n + nu$ and  $\sc{σ} = 1 + 4n + 1 + \sum_{i<n} 3 = 2 + 7n$.

\ifextended
\item Similarly, folding a $\ntext{cat} : \tlist{p}{X} ⊸ \tlist{q}{X} ⊸ \tlist{p+q}{X}$ operation with potentials $\sc{τ} = 1 + p$ and $\sc{σ} = 1$ demands instantiating $A$ as $\tlist{u}{X}$ and $B_i$ as $\tlist{u+iv}{X}$ and defines a function of type $\tlist{n}{(\tlist{u}{X})} ⊸ \tlist{v}{X} ⊸ \tlist{v+nu}{X}$ with potentials $\sc{τ} = 2 + 5n + nu$ and $\sc{σ} = 2 + 7n$.
\fi

\item Folding a $\ntext{mul} : \tnat{p} ⊸ \tnat{q} ⊸ \tnat{pq}$ operation with potential $\sc{τ} = 2pq + 2p + q + 2$ and $\sc{σ} = pq + p + 1$ demands instantiating $A$ as $\tnat{u}$ and $B_i$ as $\tnat{vu^i}$ and defines an operation of type $\tlist{n}{\tnat{u}} ⊸ \tnat{v} ⊸ \tnat{vu^n}$ with potentials $\sc{τ} = 2 + 4n + \sum_{i<n} (2vu^{i+1} + 2u + vu^i + 2) = 2 + 6n + 2nu + v(u^n + 2u^{n+1} - 3)$ and $σ = 1 + 4n + 1 + \sum_{i<n} (2 + 2vu^{i+1} + 3u + 1) = 2 + 7n + 3nu + v(u^n - 1)$, which are exponential in the size of the list.

\end{itemize}

\subparagraph{Example of a concrete application to program complexity certification}
The net built by providing interaction-net representations of a list of length $n$ containing natural numbers bounded by $u$ and an initial value bounded by $v$ (both have null sequential time potentials) as arguments to the net that folds multiplication has combined time potential $\sc{τ} = 2 + 6n + 2nu + v(u^n + 2u^{n+1} - 3)$.
According to Corollary~\ref{c:sequential-time}, this constitutes a concrete bound on the sequential time needed to perform this operation as a function of the sizes of its arguments.

\subparagraph{An analysis of merge sort}
Assuming $A$ is replicable, we denote by $\tcomp{n}{A} = \oc_n {(A ⊸ A ⊸ \tbool)}$ the type of comparison functions which can be called $n$ times.
Let $f, g, h: ℕ → ℕ$ be recursively defined as $f(n) = n - 1 + f(\fh n) + f(\ch n)$, $g(n) = 1 + g(\fh n) + g(\ch n)$ and $h(n) = n + h(\fh n) + h(\ch n)$ if $n > 1$ or null otherwise.
In particular $g(n) = n-1$ and
$f(n) = h(n) - g(n) ≤ n \log_2 n$ for $n > 0$.
We consider the following nodes:
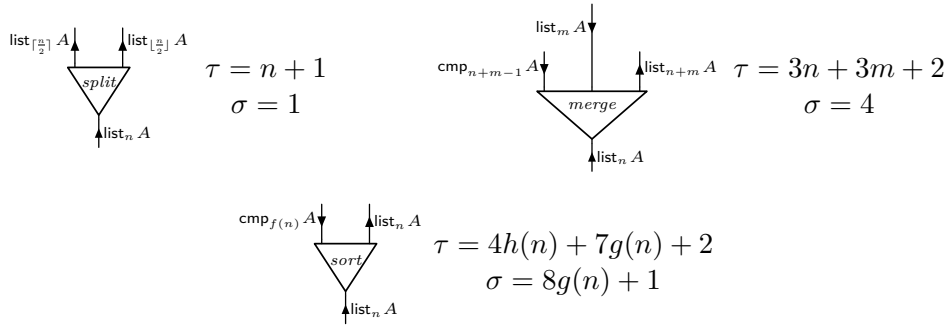
\display{
  \cscheme{\sc{τ} = n + 1 \cr \sc{σ} = 1}{
\input{asy/2c70ff9a/\@netformat-\netstyle}
  }
  \and
  \cscheme{\sc{τ} = 3n + 3m + 2 \cr \sc{σ} = 4}{
\input{asy/1f1fc76e/\@netformat-\netstyle}
  }
  \also
  \cscheme{\sc{τ} = 4h(n) + 7g(n) + 2 \cr \sc{σ} = 8g(n) + 1}{
\input{asy/52536146/\@netformat-\netstyle}
  }
}

Reduction rules for \emph{merge sort} are as follows (we display them as nested patterns \cite{interaction-nets-with-nested-pattern-matching} to avoid the explicit introduction of auxiliary nodes):
\display{
  \opreducts{
\input{asy/030a98a1/\@netformat-\netstyle}
    \and
\input{asy/63d6939f/\@netformat-\netstyle}
  }
  \and
  \opreducts[2]{
\input{asy/9bf98d4a/\@netformat-\netstyle}
    \and
\input{asy/29c3d21d/\@netformat-\netstyle}
  }
  \opreducts[2]{
\input{asy/04a1cce1/\@netformat-\netstyle}
    \and
\input{asy/8c3d1968/\@netformat-\netstyle}
  }
}

\ifextended
Implementations of $\ntext{split}$ can be obtained as follows:
\display{
  \opreducts{
\input{asy/59fcb8d0/\@netformat-\netstyle}
    \and
\input{asy/5282f268/\@netformat-\netstyle}
  }
  \and
  \opreducts{
\input{asy/6c1fcc86/\@netformat-\netstyle}
    \and
\input{asy/18f92a36/\@netformat-\netstyle}
  }
}

\fi

Potential-decrease properties follow from the recursive definitions of $g$ and $h$.

\section{Related Work}
\label{s:related}

\subparagraph{Complexity Analysis of Interaction Nets}
Despite the conceptual importance of interaction nets, little is known about complexity analysis of reductions in interaction nets, Perrinel's work~\cite{Perrinel:2014} on context semantics of interaction nets~\cite{GAL:1992} being the exception.
Perrinel employs context semantics to assign a global weight $W$ to a net $N$ which bounds the length of sequential reductions.
This is obtained analogously to (and is actually based on)~\cite{DalLago:2009}.
In contrast to our Corollary~\ref{c:sequential-time} this approach is not compositional; it requires an in-depth context semantic analysis of the net $N$ for a given input (equivalent to execution) and does not offer a generic bound.

\subparagraph{Implicit Computational Complexity and Program Analysis}
In \emph{implicit computational complexity} and in other program-analysis areas of research, space–time complexity functions are seldom used directly.
We highlight \cite{BL:CSL:12}, in which a type-based termination argument is fused together with a semantic size argument in order to classify polytime-computable higher-order functional programs.
\emph{Space–time weights} are defined which are conceptually related to the potential functions $\sc{γ}_c$.
This result is comparable to our Corollary~\ref{c:sequential-time}, although restricted to  polynomial time bounds. On the other hand, no results for space analysis or parallel reduction  have been provided, like in our work.
%
With respect to program analysis, we highlight work by Brockschmidt et al.~\cite{BEFFG:2014} in which an alternating \emph{size} and \emph{time complexity} analysis of integer transition systems is made explicit.
Similar ideas are exploited in~\cite{SZV:2014}.
The focus is on automation for first-order systems.
Automation of time complexity analysis of higher-order programs has for example been considered in~\cite{ADM:2015}.
No results for space analysis or parallel reduction have been provided.
We will consider automation of our method in future work.

\subparagraph{Sized Types and Amortized Cost Analysis}
\emph{Sized types} is a streamlined notion related to \emph{linear dependent types}~\cite{DP:2014}.
In contrast to the analysis provided by Vasconcelos~\cite{Vasconcelos:2008} our work is purely focused on user-defined annotations with sized types and additionally employs potentials based on the size annotation, in the spirit of~\cite{Atkey:2011}.
This is also different from \emph{amortized cost analysis}, where potentials are defined on values and no previous size annotation is required.
For example, Hoffmann and Shao~\cite{HS:2015} provide an extension of earlier multivariate amortized cost analyses \cite{HAH12} to parallel reductions.
The costs of parallel executions are essentially approximated using cost-free metrices in addition to a size change analysis.
While sized types and amortized cost analysis give rise to powerful automated techniques, the method proposed in this paper is more versatile and allows to incorporate non-local size analysis.


\subparagraph{Concurrency and Scheduling}
Ghica and Smith provide in~\cite{GS:2014} a generalization of bounded linear logic~\cite{GSS:1992} to bounded linear types over a semiring, which equipped with a suitable SMT solver yields a straightforward automation of type inference.
The concept of \emph{co-monadic} resource sensitivity is emphasized. It amounts to a shift in paradigm.
Instead of precisely witnessing resource use through the type system, as we have advocated here, the focus is on defining scheduling of computational tasks, based on the precise typing.

\section{Conclusion}
\label{s:conclusion}

Because input-focused complexity analysis is not compositional, analyzing properties of outputs is a necessary complement to analyzing time or space requirements.
Size annotations can be added to types and validated by a suitable typing of reduction rules.
For parallel reductions, schedule-bound transmission of (partial) data composes easily and ensures productivity: timing assumptions on inputs entail timing guaranties on outputs.
From there, we straightforwardly obtained sequential and parallel complexity bounds by assigning potentials to typed interaction-net nodes.


In particular, complexity analysis of parallel reductions may be used to optimize the dispatch and scheduling of computation tasks in distributed environments.
In general, we expect that the present resource analysis will prove useful in the design and implementation of modern hardware solutions.

To conclude, inspired by bounded linear logic \cite{GSS:1992}, we were able to extend the complexity analysis to interesting higher-order programs.
As it turns out, \emph{merge sort} can be typed $\oc_{n \log_2 n} (A \mathbin{⊸} A \mathbin{⊸} \mathsf{bool}) \mathbin{⊸} \mathsf{list}_n\,A \mathbin{⊸} \mathsf{list}_n\,A$, in which the index on the linear logic modality ensures a concrete $n \log_2 n$ bound on the number of calls to the comparison function.
Using a refined version of typing for nets \cite{the-structure-of-interaction--s-gimenez+g-moser}, we plan to extend our computational complexity results for higher-order programs to non-weak and optimal sequential reductions as well as to parallel reductions.


\todo{remove the bibtex hack for urls}

\printbibliography

\end{document}